%% file: main.tex
\documentclass[12pt]{article}
\usepackage[a4paper,
            bindingoffset=0.2in,
            left=1in,
            right=1in,
            top=1in,
            bottom=1in,
            footskip=.25in]{geometry}

\input{macro}

\begin{document}

\maketitle
 \begin{abstract}
We present the first formal treatment of \emph{yield tokenization}, a mechanism that decomposes yield-bearing assets into principal and yield components to facilitate risk transfer and price discovery in decentralized finance (DeFi). We propose a model that characterizes yield token dynamics using stochastic differential equations. We derive a no-arbitrage pricing framework for yield tokens, enabling their use in hedging future yield volatility and managing interest rate risk in decentralized lending pools. Taking DeFi lending as our focus, we show how both borrowers and lenders can use yield tokens to achieve optimal hedging outcomes and mitigate exposure to adversarial interest rate manipulation. Furthermore, we design automated market makers (AMMs) that incorporate a menu of bonding curves to aggregate liquidity from participants with heterogeneous risk preferences. This leads to an efficient and incentive-compatible mechanism for trading yield tokens and yield futures. Building on these foundations, we propose a modular \textit{fixed-rate} lending protocol that synthesizes on-chain yield token markets and lending pools, enabling robust interest rate discovery and enhancing capital efficiency. Our work provides the theoretical underpinnings for risk management and fixed-income infrastructure in DeFi, offering practical mechanisms for stable and sustainable yield markets.
\end{abstract}

\maketitle

\section{Introduction}
\label{sec:intro}
\input{intro}

\section{Pricing yield tokens}
\label{sec:model}
\input{model}

\section{Yield risk management}
\label{sec:welfare}
\input{lending}



\section{Conclusion and Future Work}
\label{sec:conclusion}
\input{conclusion}

\input{acknowledgements}

\bibliography{references}

\appendix
\input{appendix}



\end{document}

%% file: macro.tex
\title{Split the Yield, Share the Risk: Pricing, Hedging and Fixed rates in DeFi} 

\author{
  Viraj Nadkarni
  \and
  Pramod Viswanath
}

\usepackage{algorithm}
\usepackage{algpseudocode}
\usepackage{graphicx,subcaption}

\usepackage[utf8]{inputenc} 
\usepackage[T1]{fontenc}    
\usepackage{hyperref}       
\usepackage{url}            
\usepackage{booktabs}       
\usepackage{amsfonts}       
\usepackage{nicefrac}       
\usepackage{microtype}      
\usepackage{xcolor}         
\usepackage{amsmath}
\usepackage{amsthm}

\usepackage{prettyref}
\usepackage{xspace}
\usepackage{bbm}
\usepackage{mathtools}
\newcommand{\expect}{\operatorname{\mathbb{E}}\expectarg}
\DeclarePairedDelimiterX{\expectarg}[1]{[}{]}{%
  \ifnum\currentgrouptype=16 \else\begingroup\fi
  \activatebar#1
  \ifnum\currentgrouptype=16 \else\endgroup\fi
}
\newcommand{\expectGirsanov}{\operatorname{\mathbb{E}^*}\expectargGirsanov}
\DeclarePairedDelimiterX{\expectargGirsanov}[1]{[}{]}{%
  \ifnum\currentgrouptype=16 \else\begingroup\fi
  \activatebar#1
  \ifnum\currentgrouptype=16 \else\endgroup\fi
}

\newcommand{\innermid}{\nonscript\;\delimsize\vert\nonscript\;}
\newcommand{\activatebar}{%
  \begingroup\lccode`\~=`\|
  \lowercase{\endgroup\let~}\innermid 
  \mathcode`|=\string"8000
}

\newrefformat{fig}{Figure~[\ref{#1}]}

\newrefformat{app}{Appendix~[\ref{#1}]}
\newrefformat{alg}{Algorithm~[\ref{#1}]}
\newrefformat{cor}{Corollary~[\ref{#1}]}
\newrefformat{eq}{Equation~[\ref{#1}]}
\newrefformat{def}{Definition~[\ref{#1}]}
\newrefformat{as}{Assumption~[\ref{#1}]}
\newrefformat{lem}{Lemma~[\ref{#1}]}

\newtheorem{assumption}{Assumption}
\newtheorem{definition}{Definition}
\newtheorem{theorem}{Theorem}
\newtheorem{corollary}{Corollary}
\newtheorem{lemma}{Lemma}

\newcommand{\yieldBearingToken}{\ensuremath{\mathcal{A}}\xspace}

\newcommand{\yieldTokenizer}{\ensuremath{\mathcal{T}}\xspace}

\newcommand{\yieldToken}{\ensuremath{\mathcal{Y}}\xspace}

\newcommand{\principalToken}{\ensuremath{\mathcal{P}}\xspace}

\newcommand{\numeraire}{\ensuremath{\mathcal{N}}\xspace}

\newcommand{\priceYT}{\ensuremath{P_{\mathcal{Y}}^T}\xspace}

\newcommand{\lender}{\ensuremath{\mathcal{L}}\xspace}
\newcommand{\borrower}{\ensuremath{\mathcal{B}}\xspace}
\newcommand{\utilLender}{\ensuremath{U_\mathcal{L}}\xspace}
\newcommand{\utilBorrower}{\ensuremath{U_\mathcal{B}}\xspace}
\newcommand{\deltaLender}{\ensuremath{\delta_\mathcal{L}}\xspace}
\newcommand{\deltaBorrower}{\ensuremath{\delta_\mathcal{B}}\xspace}

%% file: intro.tex

The beauty of finance lies in the facilitation of price discovery in spot and derivative markets — the former creates public measures of present scarcity and enables efficient and decentralized resource allocation in an economy, and the latter creates public measures of future scarcity and uncertainty, thus aiding in decentralized planning for the future in an economy \cite{schar2021defi}. Decentralized Finance (DeFi) seeks to do the same, but in a more permissionless, transparent, disintermediated, and programmable manner \cite{werner2021sok, carapella2022defi, cfainstitute2022defi}. In particular, an important feature of DeFi has been programmable \textit{tokenization}. This property, similar to \textit{financialization} in a traditional setting, has the potential to enable efficient risk allocation and the ability to help create more stable and diversified products for end-users \cite{gudgeon2020defi, axios2022centrifuge}. In this work, we explicate how this might be done through a formal analysis of \textit{yield tokenization}, that enables the pricing of future yield and helps in managing and hedging the risks of yield volatility \cite{chen2020defiprim}.

\subparagraph{Why focus on yield?} Decentralized Finance (DeFi) has enabled users to take on a variety of roles that earn yield in return for a service. \textit{Staking} and \textit{restaking} enable users to secure and validate the underlying blockchain or decentralized application to earn rewards \cite{gudgeon2020defi}. \textit{Liquidity provision} allows users to make markets and earn trading fees in return \cite{gudgeon2020defi}. \textit{Lending} enables users to earn interest by providing borrowers access to their liquidity \cite{gudgeon2020defi}. Blockchains such as Ethereum have introduced token standards for yield-bearing assets and pools of multiple such assets called \textit{vaults} \cite{eip4626}, where users can deposit their funds to be potentially managed by curators in a mix of all the actions listed above. Each of these actions mints the user \textit{yield-bearing} tokens to represent their positions and earn the respective yields. Collectively, these activities account for nearly \$52 billion in on-chain liquidity \cite{gudgeon2020defi}, making risk management for these funds essential for sustaining and further spreading the use of DeFi by attracting more passive (or ``retail'') liquidity.

\subparagraph{Yield tokenization} The value of a yield-bearing asset can be split into the initial investment or \textit{principal} and the variable \textit{yield} that is earned over time. Each of these two components can be represented as a token on an underlying blockchain. We call this process \textit{yield tokenization}. This mechanism underpins several existing protocols such as Pendle \cite{pendle_docs}, Element \cite{element_paper}, Spectra \cite{spectra_docs}, and Notional \cite{notional_blog}. This practice is analogous to traditional financial instruments like Stripped Treasury Bonds \cite{investopedia_strips}, Collateralized Mortgage Obligations \cite{investopedia_cmo}, and Interest Rate Swaps \cite{investopedia_irs}.

\subparagraph*{Holders and speculators} Yield-bearing tokens often have a volatile rate of return. Therefore, \textit{holders} of such tokens seek opportunities to ``lock in'' that rate at a specific value in exchange for not having to worry about yield volatility. On the other hand, \textit{speculators} might want exposure to the yield because they believe they have some esoteric knowledge not priced into the market at that moment. However, they might not have the capital to buy the entire yield-bearing token without expensive leverage. Yield tokenization helps make a market for such holders and speculators. Holders can tokenize their yield and sell it on the market, thereby freezing their yield for some time period \cite{gudgeon2020defi}. Speculators can buy the yield tokens if they think the yield is going to rise in the same time period. Note that the price of yield tokens is of the order of the absolute return obtained on holding a token. Thus, the price at which speculators place their bets is much cheaper than buying the entire underlying yield-bearing token itself \cite{gudgeon2020defi}.

\subparagraph*{Risk management in lending} In this work, we primarily analyze the utility of yield tokenization for lending pools. In decentralized lending pools, the \textit{lenders} who make up the pool serve as holders of the yield-bearing \textit{lending token} \cite{gudgeon2020defi}. Speculators might be agents who have knowledge about how the lending pool is going to evolve into the future and seek to earn a premium on that knowledge. In this case, however, there are other agents who might want to use the yield token—namely, \textit{borrowers}. Borrowers are already exposed to the yield in a negative way—they have to pay the yield amount to the lending pool. Thus, they might want to hedge against a rise in future yield. To do that, they can buy the yield token and effectively ``lock in'' the interest amount they pay in a specific time period. This also helps both lenders and borrowers hedge against interest rate manipulation by an adversary \cite{gudgeon2020defi}. We make these intuitions more precise in \prettyref{sec:welfare}.

\subparagraph{Fixed interest rates} A missing link in DeFi while attracting passive liquidity has been the lack of fixed interest rate lending protocols. Although many existing protocols offer fixed rates (see, for e.g., Fluid \cite{fluid2023whitepaper}, Aave \cite{gudgeon2020defi}, Supernova \cite{gudgeon2020defi}), it is not clear if there is an active interest rate discovery mechanism that gives rise to those rates \cite{gudgeon2020defi}. This is due to a lack of market mechanisms to enable that discovery. However, yield tokenization and the advent of customizable market makers such as Uniswap v3/v4 \cite{gudgeon2020defi} allow us to design market makers on which interest rate discovery can happen efficiently and with sustainable access to liquidity. We present one such design in \prettyref{sec:fixed}.

\subparagraph*{Risk management for other yield-bearing assets} Although we focus on DeFi lending as the primary case study of where yield tokenization can be used, our results can be extended to other yield-bearing actions such as AMMs and staking as well (\prettyref{app:lsd}). In automated market makers, liquidity providers (LPs) contribute assets to a pool, earning fees from trades that constitute the yield on their LP tokens. By tokenizing this yield, LPs can sell it to secure a steady cash flow. A more powerful strategy is hedging impermanent loss (\textit{loss-vs-holding}) \cite{barrons2025impermanent}, by connecting yield tokenization to a demand lending pool \cite{barrons2025impermanent, chitra2025perpetual}, allowing LPs to borrow against their LP positions and shift the risk of asset price fluctuation to the lending pool manager. Alternatively, the liquid markets of yield tokens based on LP positions may enable estimation of arbitrage loss (\textit{loss-vs-rebalancing}) \cite{milionis2023automated}, and hedge against uncertainty in the earning of swap fees to compensate for this loss \cite{milionis2023automated}. In liquid staking, tokenizing staking yield similarly allows risk management, though it also introduces some negative externalities (\prettyref{app:lsd}). A liquid market for such yield tokens enables adversarial validators to profit by shorting the yield and misbehaving in consensus protocols, exacerbating the principal-agent problem that is present there \cite{coinflare2025liquid}. Nonetheless, yield tokenization enables hedging against other risks, such as Ethereum's congestion fees: since variable staking rewards depend on transaction/blob fees and validators are known in advance, buyers of yield tokens can hedge against congestion costs while delegators use them to lock in staking rewards \cite{fireblocks2025liquid}. This dual-sided risk transfer resembles the borrower-lender dynamics that we analyze in \prettyref{sec:amm} and can be further improved with AMM designs tailored to participants' risk preferences.

\subparagraph{Our contributions} In this work, we make the following key contributions. Firstly, we propose a mathematical model of yield tokenization using stochastic differential equations (\prettyref{sec:model1}). To our knowledge, this is the first formal treatment of yield tokenization and can be used to characterize any yield-bearing token. Secondly, given this model, we derive the fair price of a yield token from a trader's perspective, in a sense that ensures no arbitrage (\prettyref{sec:pricing}). Thirdly, we show how lenders and borrowers in DeFi lending protocols can use yield tokenization to hedge against interest rate uncertainty - thereby improving the welfare of all lending pool participants (\prettyref{sec:welfare}). Fourthly, we design a market maker (using a menu of bonding curves) for yield tokens to be traded efficiently, so that the hedging that we mentioned above is indeed tractable and all available liquidity is aggregated (\prettyref{sec:amm}). Finally, we present a design of a fixed rate lending protocol that combines the lending pool and the yield token AMM to provide quotes of the best available interest rates to lenders and borrowers, effectively acting as a lending aggregator (\prettyref{sec:fixed}). This design can be extended to lending vaults (not only pools) managed by curators as well \cite{morpho2023interest,chitra2025curationary}. We provide additional results on the pricing yield tokens for liquid staking/restaking in \prettyref{app:lsd} and give a preliminary design on how these can give rise to slashing insurance markets. All mathematical proofs are also deferred to the appendices.

\subparagraph{Related work} Although yield tokenization has not been formally treated, many current and past decentralized applications have implemented it with risk transfer as the basic motivation \cite{pendle2021whitepaper, element2021whitepaper, spectra2023docs, notional2021whitepaper, yearn2021whitepaper}. The fair pricing of yield-bearing tokens has also been gaining interest recently, with results in pricing liquidity provider tokens and associated losses \cite{yieldspace2020paper, replicating2021arxiv,nezlobin2025lossversusrebalancingdeterministicgeneralizedblocktimes}, and pricing lending tokens \cite{agilerate2024arxiv}. Our results on fair pricing of yield tokens plug into this and related literature on the models of yield in market making \cite{evans2020liquidity, replicating2021arxiv, nadkarni2024adaptive} and lending \cite{agilerate2024arxiv}. In particular, our paper addresses the risk of adversarial manipulation and interest rate volatility raised in the lending literature \cite{agilerate2024arxiv, chitra2023attacksdynamicdefirate, chitra2025curationary}. Market maker design for specific token models has been explored for prediction markets \cite{marketdesign2025paper, moallemi2024pmamm}, yield-bearing tokens \cite{pendle2021whitepaper}, MEV prevention \cite{mev2023springer}, and replicating payoffs \cite{replicating2021arxiv}. Fixed-rate lending by combining the liquidity of AMMs and lending pools is a part of many decentralized applications such as Fluid \cite{fluid2023whitepaper} and Tenor \cite{tenor2024docs}.

%% file: model.tex
In this section, we describe our formal model for the yield tokenization protocol. We give some preliminary results on how the output tokens from such a protocol may be priced in a fair manner. Note that our model is in continuous time, that we denote by $t$, but includes the discrete time setting (for yield payments corresponding to blocks in an underlying blockchain, for instance) as a special case.

\subsection{Model}\label{sec:model1}

\begin{definition}[Yield Bearing Token]
    Assume that we have a fungible token \yieldBearingToken. We call such a token \textit{yield-bearing} if the holder of such a token accrues a reward paid either in \yieldBearingToken, or in some underlying numeraire \numeraire, or in a mix of tokens $\{\mathcal{A}_i\}$.
\end{definition}

\begin{definition}[Yield Function]
    The continuously accrued yield obtained from holding the token \yieldBearingToken is given by the \textit{yield function} $\bar{Y}(t,\bar{X_t})$ where $\{\bar{X_t}\}$ denotes the vector of price processes of some underlying tokens upon which the yield depends.
\end{definition}

 For instance, \yieldBearingToken may denote a lend position in a lending protocol such as Aave \cite{aaveGov}. In that case, the interest accrued $\bar{Y}(t,\bar{X_t})$ (i.e. the yield) may change with time, and depend on the prices of the asset being lent out and the collateral being posted by borrowers. Also the yield would be a scalar quantity. Alternatively, \yieldBearingToken may denote the liquidity provisioning (LP) token in an automated market maker such as Uniswap \cite{uniswap2021v3}. In that case, the fees earned serve as the yield vector $\bar{Y}$, while $\bar{X_t}$ would be the price processes of all assets provided as liquidity to the market maker.

\begin{definition}[Yield Tokenizer]
    A yield tokenizer \yieldTokenizer is a smart contract that allows a holder to deposit \yieldBearingToken into the contract for a time-to-maturity $T$ in exchange for minting a \textit{principal token} \principalToken and a \textit{yield token} \yieldToken.
\end{definition}

\begin{definition}[Principal Token]
    A principal token \principalToken gives its holder the right to collect the underlying yield bearing token \yieldBearingToken at the time-to-maturity $T$.
\end{definition}

\begin{definition}[Yield Token]
    The yield token \yieldToken gives its holder the right to collect all future yield accrued on the underlying token \yieldBearingToken, until such a time when either the holder decides to sell or send \yieldToken to another holder, or the time-to-maturity $T$ is reached.
\end{definition}

\begin{definition}[Yield Future]
    The yield future $\yieldToken_{t,t+\delta}$ gives its holder the right to collect the future yield between a time $t$ and $t+\delta$ accrued on the underlying token \yieldBearingToken. 
\end{definition}

A single yield \textit{token} can thus be split into multiple yield \textit{futures} representing payments until maturity. 

We make some further financial modeling assumptions as follows.

\subparagraph*{Risk free rate} We express all asset prices in a numeraire \numeraire, unless specified otherwise. For a holder of \numeraire, we assume that it is possible to accrue a risk-free yield at a continuously compounded rate $r(t)$. For example, if \numeraire is the ETH token, then the action that accrues a risk free rate $r(t)$ would be Ethereum native staking \cite{ethereum_staking}. If \numeraire is the US Dollar, then the LIBOR \cite{libor_bba} would act as the risk free rate.

\subparagraph*{Price processes} The prices of all underlying assets of the yield bearing token \yieldBearingToken change with time in a stochastic manner. We model this change via the following vector stochastic differential equation, 
\begin{align}\label{eq:price_main}
    d\bar{X_t} = \bar{\mu_t}(\bar{X_t}) dt + \Sigma_t(\bar{X_t}) d\bar{W_t}
\end{align}
where $\bar{X_t},\bar{\mu_t} \in \mathbb{R}^n$, $\bar{W_t} \in \mathbb{R}^m$, and $\Sigma_t \in \mathbb{R}^{n \times m}$. Here, $\bar{W_t}$ denotes a vector of independent Brownian motion processes. This implies that the entries of the random vector $\bar{X_t}$ are correlated in a manner dictated by the matrix $\Sigma_t$. For notational ease, we drop the dependence of $\bar{\mu_t}$ and $\Sigma_t$ on $\bar{X_t}$ going forward.

\subsection{Fair pricing}\label{sec:pricing}

We have formalized the model of how a yield bearing token may have its yield tokenized. We also saw, in \prettyref{sec:intro}, how the market for these yield tokens is formed and the motivations behind different traders participating in it. Now, we derive how such a yield token may be fairly priced in this market, from the perspective of a potential trader who has a specific model of how prices of the underlying assets are going to evolve into the future.

\subparagraph*{Fair pricing} Every trader who wishes to buy or sell the yield token may have their own model of how the prices of the underlying tokens are changing and hence how much yield $\bar{Y}(t,x)$ will be generated in the future. For such a trader, the fair price of the token would be the price at which there can be \textit{no risk-free profit} made given how they think the prices are going to evolve stochastically. Thus, we apply the canonical no arbitrage condition that is used in finance to price options, to also price a yield token. The method is simple - we form a portfolio of assets so that the stochastic part of the portfolio value is canceled out, and then equate the value gained in a small time with the value gained had the portfolio been denominated in numeraire and in a risk-free investment. This gives us a partial differential equation in how the price of the yield token should evolve. The price can then be calculated (at least numerically in most cases, if not analytically) given the terminal conditions (in this case, we use the fact that the yield token is worth nothing at maturity).  

\subparagraph*{Price of the yield token} Suppose that a trader knows that the prices of the assets underlying the yield token evolve as per \prettyref{eq:price_main}. Let the price of the yield token \yieldToken with maturity $T$, at time $t$, be denoted by $\priceYT(t,\bar{X_t})$. Then, we have the following result.

\begin{theorem}\label{thm:1}
    The price of a yield token \yieldToken with maturity $T$ minted from a yield bearing token \yieldBearingToken with underlying tokens $\{ \yieldBearingToken_i\}_{i\leq n}$ is given by
    \begin{align}\label{eq:yt_price}
        \priceYT(t,\bar{X_t}) &= \expectGirsanov*{\int_t^T e^{-\int_t^s r(u) du} \bar{Y}(s,\bar{X_s})\cdot \bar{X_s} ds | \bar{X_t}}
    \end{align}
    where the random process $\bar{X_s}$ follows the stochastic differential equation
    \begin{align}\
        d\bar{X_t} = r(t)\bar{X_t} dt + \Sigma_t d\bar{W_t^*},
    \end{align}
    where $r(t)$ denotes the risk free rate for the chosen numeraire \numeraire, and $\bar{W_t^*} \in \mathbb{R}^m$ is standard vector Brownian motion with independent entries.
\end{theorem}

\subparagraph*{No arbitrage across maturities} Although the yield token at a specific maturity $T$ is fungible, it is still distinct from a yield token at some other maturities $T'$. The same underlying yield bearing token can give rise to a plethora of yield tokens, each with its own exchanges and liquidity. However, the prices of these seemingly distinct yield tokens are clearly correlated, and one can take advantage of this correlation and unify liquidity across differing maturities. The expression we derived in \prettyref{thm:1} gives the price of a yield token for a specific maturity $T$. We claim that calculating the price of the yield token across different maturities gives prices that are \textit{consistent}, in the sense that they do not present a risk-free arbitrage opportunity at any point. More formally, we have the following corollary.
\begin{corollary}\label{cor:1}
    A portfolio of yield tokens with varying maturities, with the same underlying yield-bearing token, and priced according to \prettyref{thm:1} cannot have an arbitrage opportunity.
\end{corollary}

\subparagraph{Splitting a yield token into yield futures} \prettyref{cor:1} simplifies the task of unifying liquidity across maturities in a yield token market. We can go even further and express any yield token as the sum of yield tokens of smaller maturities in a natural way. For instance, buying a yield token for the time period $[0,T]$ is equivalent to buying $n$ yield tokens for the time periods $[0,t_1],[t_1,t_2],\cdots,[t_{n-1},T]$ for any $0 < t_1 <\cdots< t_{n-1}<T$. The yield tokens with the smallest time periods can thus serve as the building blocks to price yield tokens with longer maturities. For instance, the times $t_1,t_2,\cdots, t_n$ can represent the timestamps of consecutive blocks on the underlying blockchain. However, note that existing protocols do not offer a token that represents the yield in a future time period such as $[t_1,t_2]$. One can write a \textit{yield splitting} contract, however, that takes a yield token \yieldToken as a deposit and then mints a bunch of \textit{yield futures} $\{ \yieldToken_i\}$ for every yield emission in its future. Assuming the underlying yield is emitted at times $t_1,\cdots,t_{n-1}, T$, we define $Y_i = \bar{Y}(t_i,\bar{X_{t_i}})\cdot \bar{X_{t_i}}$. The random variable $Y_i$ represents the yield emitted by the underlying \yieldBearingToken at a future time $t_i$. By \prettyref{thm:1}, a trader can determine the fair price of a yield future in a straightforward manner.
\begin{corollary}
    The fair price of the yield future $\yieldToken_i$ is given by 
    \begin{align}
        P_{\yieldToken_i} = \expectGirsanov{e^{-\int_0^{t_i} r(u) du} Y_i | \bar{X_0}}
    \end{align}
    where the random process $\bar{X_s}$ follows the stochastic differential equation
    \begin{align}\
        d\bar{X_t} = r(t)\bar{X_t} dt + \Sigma_t d\bar{W_t^*},
    \end{align}
\end{corollary}

Yield tokenization protocols such as Pendle currently do not offer yield futures. However, in \prettyref{sec:amm}, we will see that yield futures help interest rate discovery and boost liquidity in yield markets. This eventually helps in optimal risk allocation in lending markets, and further build fixed rate lending on top of this (\prettyref{sec:fixed}).

\subparagraph*{Calculating implied yield from a given market of tenors} The above analyses holds for a single trader wishing to buy or sell yield tokens, given their specific belief about how the prices of the underlying tokens are going to evolve in the future. The yield token market would collate the disagreements of such traders over these beliefs. This gives rise to a set of equilibrium prices (across maturities), that \textit{imply} what the market thinks the yield will be in the future. For instance in the price of a yield token that expires at maturity $T$ is $p$, then one should expect $p = E[Y]$, where $Y$ is the present value of all future yield payments. However, this simple interpretation of yield token prices only holds if the markets for such tokens are liquid. In the presence of limited liquidity, as we shall see in \prettyref{sec:amm}, there is a gap between the fair price of a yield token and the market price due to the risk aversion of the liquidity provider and traders. Modulo these caveats, for a reasonably liquid market, the prices of yield tokens can provide a useful estimate of future yield from an underlying yield bearing asset/protocol.

\subparagraph*{Pricing in the presence of limited liquidity} Most token pairs in DeFi have limited trade volume and liquidity \cite{kpmg2021crypto}, making order books impractical for price discovery. Instead, the use of Automated Market Makers (AMMs), specifically Constant Function Market Makers (CFMMs), has enabled swapping for informed traders and has acted as price oracles for uninformed traders \cite{radix2021amms}. Yield token markets are inherently low volume, since each yield token corresponds to a specific yield-bearing application (such as a specific lending pool) for a specific maturity. This indicates the need for AMMs that can make use of the knowledge that potential liquidity providers may have about the underlying yield-bearing mechanisms (such as the supply/demand dynamics in an underlying lending protocol). We shall see, in the following sections, how AMMs may be designed to boost liquidity to enable better use of yield tokens. Specifically, we shall use the canonical concept of a \textit{bonding curve} of AMM reserves \cite{linumlabs2019bonding}, that constrains the values that the liquidity provider's inventory may take after interacting with a trader. Several AMM protocols exist that make use of bonding curves, this includes Uniswap v2 \cite{uniswap2020v2}, Balancer \cite{balancer2021amm} that use the Geometric Mean bonding curve $x^{\theta}y^{1-\theta} = \mathrm{constant}$. Following these, the concept of \textit{concentrated liquidity} was introduced in Uniswap v3 \cite{uniswap2021v3}, this enabled a liquidity provider to specify the price range in which the bonding curve can use their liquidity. More recently, customizable hook smart contracts were introduced in Uniswap v4 \cite{uniswap2023v4} that allow the liquidity provider to specify variable fees and bonding curves. We use canonical concepts from these AMMs to design our yield token market maker in \prettyref{sec:amm}.

%% file: lending.tex
 Lending pools in DeFi provide users, called \textit{borrowers}, access to liquidity at an interest by posting collateral. On the other hand, \textit{lenders} can deposit their assets in the pool, in the hope to earn that interest, and get lending tokens that represent their position. Each lending position earns a pro-rata share of interest for the holder of the lending token. This makes a lending token an instance of a yield-bearing token, with the share of the interest on the loans acting as the potentially time-varying yield. The biggest lending protocols in DeFi set their interest rates as a function of the borrow and lend positions. If the total amount of asset lent at time $t$ is $L_t$ and the total amount of asset borrowed is $B_t$, then the interest rate is set by the protocol using a piecewise linear increasing function of $B_t/L_t$. These interest rate curves are static and are only updated every quarter or so, via a governance vote \cite{bastankhah2024thinking}. The goal of these curves is to make the \textit{utilization ratio} of the pool $B_t/L_t$ stay as close to an optimal value $U^*$ as possible. More recently, dynamic lending pools have been proposed \cite{baude2025optimalriskawareratesdecentralized,bastankhah2024thinking,agilerate2024arxiv}, that take historical changes in borrow demand and lending supply into account to estimate and propose an equilibrium interest rate that achieves the desired utilization $U^*$. 

 \subparagraph{Two views of the market} We first zoom in and look at how a trader who wants to express his knowledge can use results from the previous section to make a decision about whether they should go long or short a yield token being traded. In the view of \prettyref{thm:1}, this amounts to having the right model of how interest rate in a lending pool evolves. We give several plausible candidate models that have been developed in DeFi literature recently. Each of these models can be the basis on which \prettyref{eq:price_main} can be specified with customized forms for $\bar{\mu}_t(.)$ and $\Sigma_t(.)$. After that, we zoom out and look at how the welfare of the lending pool as a whole (i.e. the borrowers and lenders) is affected by the presence of the yield tokenizer. Looking at these externalities reveals that risk is allocated in a much better way due to yield tokenization, and it also boosts capital efficiency. We will see how an AMM for a yield token may be designed to provide liquidity in a sustainable manner. This amounts to giving a liquidity provider the right tools to express her knowledge about trader behavior (in particular risk aversion). We also see how a liquid yield token markets give rise to fixed rate lending and money markets.

\subsection{Price of yield-tokenized lending}

From \prettyref{sec:pricing}, we see that in order to calculate the fair price of a yield token, one needs to have a stochastic model of how yield would evolve in the future. We give some pointers on how one might go about this, but leave the full calculation of prices with more empirical data to future work. 

\subparagraph{Interest rate models - TradFi} In traditional finance, interest rate has been modeled in the past via mean reverting processes such as the Vascek model \cite{vasicek1977equilibrium}, the Cox-Ingersoll-Ross model \cite{cox1985theory}, the Hull-White model \cite{hull1990pricing}. More recent approaches, such as the Heath–Jarrow–Morton framework \cite{heath1992bond}, or LIBOR Market Models \cite{brace1997market} use a varying yield curve with differing short-term and long-term stochasticity. These are used to price many yield derivatives such as swaps and swaptions. However, these models implicitly assume that the interest rate market is already liquid and exogenous, and that all derivatives are being marked-to-market using this external market. In what follows, we shall see that this is not true for DeFi lending pools, and thus, they require a different suite of models.

\subparagraph{Interest rate models - DeFi} In DeFi lending pools, modeling the interest rates boils down to modeling how borrowing demand and lending supply would vary with time. This is because this demand and supply, as manifested in a particular lending pool locally, are the primary determinants of how that lending pool sets interest rates. Additionally, for a specific lending pool, this demand/supply also depends on the availability and value of collateral to support the loan. Changing collateral requirements and liquidation thresholds may shift the demand/supply dynamics with time \cite{gudgeon2020defi}. Once a model for demand and supply are specified for a lent asset and a suite of collateral assets, the interest rate itself depends on a customized function that seeks to keep the \textit{utlilization ratio} (the ratio of demand and supply) of the lending pool around a constant $0< U^* \leq 1$. These factors imply that the interest rate cannot be directly modeled via the traditional models mentioned above, since it is, to a large extent, set endogenously. Recent research proposes several empirical models that focus on fitting time varying demand/supply curves on historical data \cite{baude2025optimalriskawareratesdecentralized, cohen2023economics, cohen2023paradox}. These models can then be applied on an interest rate protocol that takes demand/supply time series as an input and gives the interest rate as an output \cite{agilerate2024arxiv, bastankhah2024thinking}. Models that view a DeFi lending position as a replicating portfolio of options have also been developed \cite{halpern2024fairratesimpossiblelending}, and can help price the yield token directly in the traditional Black-Scholes setting. A lending-based yield token may also be minted on \textit{vault} deposits that allocate liquidity dynamically in more recent lending protocols such as Morpho \cite{morpho2023interest}. In those cases, the interest rate determination is done using online learning algorithms by agents called curators \cite{morpho2023interest,chitra2025curationary}, thus requiring a careful empirical study of how the choices of curators translate to interest rate discovery.

\subsection{Hedging: a positive externality}

We now look at how borrowers and lenders may use the yield token in order to hedge their positions. The idea is to get a sense of how the yield tokens interact with the incentives of the lending pool participants, and affect their overall welfare. We show that this effect is \textit{positive} - and quantify the positivity of this externality under different market conditions.  To take this macroscopic view, we first make some simplifying assumptions, and then relax some assumptions step-by-step.

\begin{assumption}\label{as:2}
    The lending pool finds the equilibrium interest rate corresponding to the demand and supply curves of the rational lenders and borrowers at time $t=0$.
\end{assumption}
This says that we treat the lending pool and how it discovers the equilibrium interest rate as an abstraction. Thus, the way it finds that rate may be using static interest rate curve, a more dynamic curve, or through a mixture of possible models. In any case, we start our analysis after the lending pool has achieved an equilibrium. 
\begin{assumption}\label{as:3}
    All lending pool participants are risk neutral.
\end{assumption}
This implies that any market participant values a random reward in the future the same as a fixed reward equal to the expectation of the random reward.
\begin{assumption}\label{as:4}
    Any holder of a yield token is either a lender or a borrower in the underlying lending pool.
\end{assumption}

\subparagraph{No arbitrage under honest lenders and borrowers} The simplest interaction between the lending pool and the yield token market is that of arbitrage. We show that, if the tenor of any yield tokenizer is small enough so as to not warrant a change in demand or supply curves, then the price of the yield token should be exactly equal to the interest rate. We do this by showing that an arbitrage opportunity opens up otherwise, which implies that the market is not in equilibrium.

\begin{lemma}\label{lem:1}
    Let $Y$ be the interest amount earned in \numeraire over the period $[0,T]$, as set by the lending pool. Then, under assumptions \prettyref{as:2}, \prettyref{as:3}, and \prettyref{as:4} the price of the yield token is given by $$\priceYT = E[Y].$$
\end{lemma}

\subparagraph{Risk averse lenders and borrowers} We now move to a more accurate description of the market participants by relaxing Assumption \prettyref{as:3}. A lender who deposits assets into a pool faces the risk of the pool interest rate dropping below what they expect, while a borrower faces the risk of the interest rate increasing. No matter what methods are used to set interest rates, a lender who has entered the pool agrees on getting paid at least the current equilibrium interest rate, while a borrower agrees to pay at most that amount. If the yield token market is in no arbitrage with the lending pool, we expect risk averse lender to sell, and borrowers to buy yield tokens so as to avoid any changes in the interest rate in the future. We formalize this intuition in the following lemma.

\begin{lemma}\label{lem:4}
Let a lender \lender  deposit $L$ amount of asset and borrower \borrower borrow $B$ amount of the same asset from a lending pool. Let $Y$ be the interest earned by the lender per unit of asset deposited. They respectively have concave utility functions \utilLender,\utilBorrower. In the presence of a yield tokenizer, we get the following changes to their behavior.
\begin{enumerate}
    \item \lender prefers to tokenize her yield and sell the yield token at any price above $E[Y] - \deltaLender$,  
    \item $\mathcal{B}$ prefers to buy yield tokens at any price below $E[Y] + \deltaBorrower$,  
\end{enumerate}
where $ \deltaLender,\deltaBorrower$ are non-negative constants that depend upon $\utilLender, \utilBorrower, L, B$ and the distribution of future yield $Y$.
\end{lemma}

We see that risk aversion implies that any lender or borrower is willing to deviate from the fair price to go lower and higher respectively. Thus, if one has the knowledge of how risk averse lenders/borrowers are, one can earn a profit as a market maker by setting prices according to \prettyref{lem:4}. We give recommendations on how this can be done in \prettyref{sec:amm}, where we discuss how effective hedging can be done when we have limited liquidity in the yield token markets.

\subparagraph{Hedging increases welfare \textit{and} protects against adversaries} The behavior of risk averse lenders and borrowers implies that, in absence of any other buyers/sellers of these yield tokens one would have a perfectly hedged market. Such a market can also avoid the negative effects of interest rate manipulation by an adversarial user, in case the underlying interest rate mechanism of the lending pool is adaptive \cite{agilerate2024arxiv}.

\begin{theorem}\label{thm:5}
Suppose we have multiple lenders $\{\lender_i\}$ with lent amounts $\{L_i\}$, and multiple borrowers $\{\borrower_i\}$ with borrowed amounts $\{B_i\}$ in a lending pool. Then, we have the following results.
\begin{enumerate}
    \item All lenders and borrowers can completely hedge their positions till time $T$, after selling/buying the yield token at a price equal to the fair price $E[Y]$.
    \item The welfare of the lending pool participants increases in presence of a yield tokenizer.
\end{enumerate}
    
\end{theorem}

\subparagraph{Fixed rate lending} For a lending pool with a static interest rate rule such as Aave or Compound \cite{aaveGov,compoundGov}, a yield tokenizer opens up a fixed interest rate loan market for a specific period of time. For example, suppose a borrower asks for an interest rate quote on borrowing an asset till some time $T$. The lending pool can calculate how much fixed interest rate can be charged, based on the current prices of the yield tokens and current variable interest rate in the lending pool. If the borrower, accepts the quote, the lending pool opens up a borrow position, sending the lent asset and the corresponding yield tokens to the borrower. This allows the borrower to have an effectively fixed interest rate. We give a more explicit implementation of fixed rate lending in \prettyref{sec:fixed}.

\subparagraph{Role of speculators} \prettyref{thm:5} tells us that in a completely hedged lending pool, the effects of an interest rate manipulation are not felt by the participating lenders and borrowers. However, the pool still would give a suboptimal view of interest rates to \textit{potential} lenders/borrowers. In a typical lending pool, the elasticity of loan demand and supply is often quite low \cite{aave-report-2023}, and thus there is not a lot of corrective action in response to a rate manipulation. For instance, if the interest rate is increased by an adversary, many potential borrowers would be priced out until lenders put more assets into the pool and drive the interest rate down. However, pairing a lending pool with a yield tokenizer opens up a risk free arbitrage opportunity - namely, a speculator who has a view of what the real equilibrium rate is can lend a small amount of asset, sell the corresponding yield tokens. This would cause the price of yield tokens to drop, thereby forcing underlying interest rate in the pool to drop as well due to the arbitrage opportunities outlined in the proof of \prettyref{lem:4}. Such an arbitrage opportunity would exist as long as there is a significant manipulation by the adversary, thereby bringing the interest rate back to a level which is much closer to attract potential lenders and borrowers alike - the equilibirum value.

\subparagraph{Welfare in presence of speculators and the borrow-lending gap} Many lending pools present currently have a large gap between borrowing and lending interest rates. So if the yield from such pools is tokenized, borrowers would have to either bid up the prices of yield token in order to achieve complete hedging, or hedge only a part of their loan. This would decrease the welfare of the protocol users as a whole, since risk averse borrowers would only be able to partially hedge their positions. If the borrow absolute interest payment is $\gamma>1$ times the lender's payment, then we get a worse welfare than complete hedging. More formally, we have the following corollary to \prettyref{thm:5}.
\begin{corollary}\label{cor:7}
    In the scenario of \prettyref{thm:5}, if the borrowers have to pay a multiplicative premium of $\gamma>1$ on their interest rate, then the welfare of the lending pool participants is non-increasing in $\gamma$.
\end{corollary}

\subparagraph{Increase in market efficiency} In previous work \cite{chitra2023attacksdynamicdefirate}, the main antidote to adversarial interest rate manipulation has been proposed to be a gap between borrower and lender interest rates. This makes interest rate manipulation less profitable, but also decreases capital efficiency. Thus, an adaptive lending pool has to trade off capital efficiency for interest rate credibility \cite{agilerate2024arxiv}. The presence of a yield tokenizer can help resolve this dilemma. To see this, note that a wider gap makes hedging using yield tokens \textit{less} effective. This is because a borrower has to buy a yield token at a higher price than the yield being paid as a holder of the token. This would hedge part of the interest rate payment, but would still leave the borrower exposed to some part of the rate manipulation or volatility risk. Thus, if a lending pool is paired with a yield tokenizer, \prettyref{thm:5} encourages the pool protocol to narrow that gap. It guarantees that doing so would improve user experience (since we have worse user welfare as shown in \prettyref{cor:7}), while simultaneously protecting the lending pool participants against adversarial manipulation of rates, and volatility.

\section{Designing market makers for yield risk management}\label{sec:amm}

In \prettyref{sec:welfare}, we assumed that the market of lenders and borrowers is already in static equilibrium, where all lenders have lent their assets and sold their tokenized position at the same time as the borrowers opening their positions and buying the yield token (Assumption \prettyref{as:4}). In this section, we analyze what happens when we relax this assumption as well. This is a more realistic scenario - a lender may choose to lend their assets at any time and sell yield tokens, while a borrower may choose to borrow and buy the yield tokens at any other time. There might be other market participants such as speculators or liquidity providers that help make the market. This calls for a market maker that can effectively transfer interest rate risk from the lenders to the borrowers via the yield tokens. 

\subparagraph{Current AMM design is adhoc} The way yield token AMM pools operate right now is according to a bonding curve that varies as per a schedule set by the protocol that is managing the pool \cite{pendle_amm_docs}. While the protocol might have some view of how yield is going to vary in the future, the LP's preferences are not taken into account. An ideal market making mechanism should find a way to aggregate the liquidity of LPs across different risk appetites so that price discovery happens in a bottom-up manner, and liquidity allocation is not enforced top-down via one particular bonding curve. In what follows, we make this intuition more precise - we will see that using one particular bonding curve schedule limits the liquidity that is attracted to the market maker. On the other hand, if we allow the LPs to choose their bonding curve as per their risk appetite, and aggregate the liquidity across these pools, it would attract far more liquidity. This makes the process of risk allocation from the lenders to the borrowers much smoother. Another problem with the current design is the fragmentation of liquidity across yield token maturities. We see how the introduction of \textit{yield future tokens} can help unite liquidity.

\subparagraph{Liquidity provider's incentives} The ask and bid prices, at which a market of yield tokens can be made, depend on the preferences of a liquidity provider and its beliefs about how the yield will evolve in the future. Only if the prices set by the liquidity provider appeal to the existing risk averse lenders and borrowers, will we have an efficient transfer of risk. We now seek to determine how a liquidity provider may set prices for yield tokens via an automated market maker pool. Let us assume that a liquidity provider has a utility function $U_P(w)$ over their wealth $w$ denominated in the numeraire \numeraire. A risk averse liquidity provider would have a concave and increasing $U_P(.)$. Also assume that, initially, the LP has amount $y$ denominated in the yield token and amount $x$ in the numeraire. First, let us consider the simple case of a yield token that has only one yield payment due before maturity. This would be equivalent to a yield future that we define in \prettyref{sec:model1}.

\subparagraph{Optimally efficient bonding curve} For an LP as defined above, we can derive a \textit{bonding curve} that the LP can use to express those preferences. We show that this curve satisfies the basic property of concavity, which implies that it can be used as a valid constant function market maker (CFMM) \cite{angeris2020does}. This curve is also \textit{optimally efficient}, in the sense that the LP does not expect any change in expected utility if the trader interacts with it as per the bonding curve. Towards this, we first prove the following result.
\begin{lemma}\label{lem:8}
    Suppose that the liquidity provider with utility $U_P$ (concave, increasing) has to buy \yieldToken from a seller. Let $p_1, p_2$ denote the maximum prices the LP is willing to pay for a quantity $\Delta_1, \Delta_2$ of the yield token respectively, where $\Delta_1 < \Delta_2$. Then, $p_1 > p_2$. 
\end{lemma}

Similarly, for a potential buyer, we can show the following result.
\begin{lemma}\label{lem:9}
    Suppose that the liquidity provider with utility $U_P$ (concave, increasing) has to sell \yieldToken to a buyer. Let $p_1, p_2$ denote the minimum prices the LP is willing to get for a quantity $\Delta_1, \Delta_2$ of the yield token respectively, where $\Delta_1 < \Delta_2$. Then, $p_1 < p_2$. 
\end{lemma}

Using these lemmas, we get the following theorem.
\begin{theorem}\label{thm:10}
    Given a liquidity provider with concave, increasing utility $U_P(.)$, and initial reserves $x_0,y_0$ of the numeraire and the yield token respectively, we define
    \begin{align}
        \psi_S(x_0,y_0) &= \{ (x,y) :  y = y_0 + \Delta, x = x_0 - \Delta p_\Delta^S, \Delta \geq 0 \}, \\
        \psi_B(x_0,y_0) &= \{ (x,y) :  y = y_0 - \Delta, x = x_0 + \Delta p_\Delta^B , \Delta \geq 0 \}, 
    \end{align}
    where,
    \begin{align}
        p^S(\Delta) &=\inf \{ p \geq 0 : \expect{U_P(x + y Y)} \leq \expect{U_P(x + yY + \Delta (Y - p))} \}, \\
        p^B(\Delta) &=\sup \{ p \geq 0 : \expect{U_P(x + y Y)} \leq \expect{U_P(x + yY + \Delta (p - Y))}\}. 
    \end{align}
    Here, $Y$ represents the random yield payment that a holder of \yieldToken gets at maturity, and all expectations are taken using the belief distribution that the liquidity provider has for $Y$.
    Then, the curve $\psi(x_0,y_0,U_P, Y) = \psi_S \cup \psi_B$ represents the optimally efficient bonding curve for the liquidity provider.
\end{theorem}

The curve $\psi$ derived above represents the frontier of all possible trades that the liquidity provider is willing to do with a buyer or seller of a yield token at any time. Note that it provides us with a map from the preferences and beliefs of the liquidity provider (the utility function $U_P$ and the probability distribution of $Y$) to the bonding curve  of permissible trades. However, this market maker has the following two shortcomings. 

\subparagraph{Preferences of traders} Firstly, the AMM only takes into account the belief of the LP over the price of the yield token, but not over the preferences of traders that are willing to sell/buy the yield token. In particular we know that a more risk averse lender is willing to tokenize her loan and sell it at a lower price than a less risk averse lender. Similarly, a more risk averse borrower is willing to pay a higher price to buy the token and hedge his interest payment than a less risk averse borrower. An LP can leverage their knowledge of how risk averse they expect borrowers and lenders to be to charge a higher premium to more risk averse traders. The way yield tokens are traded currently penalizes a less risk averse trader by giving them a worse price. This is because, for instance, a seller of yield tokens not willing to wait before selling the tokens (more risk averse) moves the price on the existing bonding curve to a lower value where the next seller (less risk averse) is supposed to sell. We show that the LP can use her knowledge of the traders' risk aversion to design a \textit{menu} of bonding curves with fees to earn premia from higher risk aversion.

\subparagraph{Multiple yield payments} Secondly, the AMM assumes that only one yield payment is due for the yield token. In reality, one may have multiple payments due before maturity. We show that the menu of bonding curves enables the trader to sell/buy to avoid/hedge some the risky yield payments, depending on his risk aversion. Further, we show how this menu can be adjusted to optimize the premium earned by the LP, thus optimally incentivizing liquidity provision. For the following analysis, we assume that the yield token \yieldToken has $n$ discrete interest rate payments due before maturity. Let these be represented by the random variables $Y_1,\cdots,Y_n$, where the payment $Y_n$ is the last one before maturity.

\subparagraph{The incentives of a yield trader} Let us now focus on the incentives of a yield trader who is either a lender or a borrower. We characterize the trader by a utility function $U_S(w)$ over his wealth $w$. Since both lenders and borrowers are risk averse, we assume $U_S$ to be concave and increasing. As discussed before, a more risk averse lender who has tokenized her loan would be willing to sell it at a lower price than a less risk averse lender. From the perspective of an LP, this implies that the trader can be shown a menu of bonding curves where each bonding curve is used based on when the trader seeks to sell/buy \yieldToken. We now make this menu more precise.

\subparagraph{Optimally efficient menu for the LP} A single AMM bonding curve can be thought of as a continuous set of quotes for different trading sizes. A trader interested in selling $\Delta$ amount of the yield token is quoted a price $p_\Delta$ based on the changes in reserves permissible under the bonding curve $\psi$. By \prettyref{thm:10}, the bonding curve $\psi$ can be worked out based on the value of yield payments remaining at the time of the trade. Therefore, if the trader decides to sell the yield token just before payment $Y_j$, the LP should quote a price based on a bonding curve determined by the distribution of $\sum_{i=j}^n Y_i$. This implies that the LP should present a menu of $n$ bonding curves, so that the trader can choose when to schedule a buy or sell order. More formally, we can arrive at an optimally efficient menu as follows.
\begin{corollary}
    For yield token \yieldToken with outstanding yield payments $Y_1,\cdots,Y_n := \{Y_i\}_1^n$, we define a menu of bonding curves to be the tuple
    \begin{align}
        \Psi(x_0,y_0,U_P,\{Y_i\}_1^n) = \left(\psi_1, \psi_2,\cdots,\psi_n\right)
    \end{align}
    where the curve $\psi_j = \psi(x_0,y_0,U_P,\sum_j^n Y_i)$ and the curve function $\psi$ is as defined in \prettyref{thm:10}. Then, each curve $\psi_j$ in the menu is optimally efficient.
\end{corollary}

The menu $\Psi$ is optimally efficient - for any trade, it quotes the prices that keep the expected future utility of the LP constant, given her risk appetite and belief over the yield. However, the LP can improve the revenue that they obtain by taking into account the risk preferences of traders (i.e. lenders and borrowers). In what follows, we take a look at the incentives of the trader, and how the LP can use her knowledge of the trader's preferences to increase her profit.

\subparagraph{The LP-trader equilibrium} Given a menu of bonding curves by the liquidity provider, let us now look at the choices of a trader. In particular, assume that the trader wishes to sell an amount $\Delta$ of the yield token. There are $n+1$ choices that he faces - either sell the tokens immediately, or sell them after the first yield payment, or sell them after the first two yield payments,$\cdots$, or sell them after the $n-1^{th}$ payment, or never sell them. Let the prices quoted to the trader by the menu of bonding curves be $p_1,p_2,\cdots,p_n$ for the first $n$ actions. Corresponding to each of the possible actions, the utility of the trader's wealth becomes
\begin{align*}
    \Tilde{U}_S^t(\Delta,p_t) &:= \expect{U_S((\sum_{i=1}^{t-1} Y_i + p_t) \Delta)} \cdots \textrm{where\ } t \in \{1,2,\cdots,n\},\\
    \Tilde{U}_S^{n+1}(\Delta) &:= \expect{U_S((\sum_{i=1}^n Y_i) \Delta)},
\end{align*}
where $\Tilde{U}_S^t(\Delta)$ represents the expected utility of the trader after choosing to perform the $t^{th}$ action. For each choice of the trader, the liquidity provider has the following utility values
\begin{align*}
    \Tilde{U}_P^t(\Delta) &:= \expect{U_P(x - p_t \Delta + y\sum_{i=1}^n Y_i+\Delta\sum_{i=t}^n Y_i)}\cdots \textrm{where\ } t \in \{1,2,\cdots,n\},\\
    \Tilde{U}_P^{n+1}(\Delta) &:= \expect{U_P(x + y\sum_{i=1}^n Y_i )}.
\end{align*}
Thus, we have two sets of equations - the first set provides us the trader's perspective by delineating the actions that he is maximizing over (namely, when to sell the yield token), and the second set provides us the liquidity provider's perspective by delineating the actions that she is maximizing over (namely, the menu of prices to be presented to the trader). Since both the trader and LP are maximizing their own utility, the equilibrium can be characterized by the following optimization problem.
\begin{definition}The equilibrium between the yield token seller and the liquidity provider is given by the tuple $(t^*, p_1, p_2, \cdots, p_n)$ such that
    \begin{align}
    t^* &= \arg\max_{t\in \{1,\cdots,n+1\}} \Tilde{U}_S^t(\Delta,p_t)\label{eq:equil}\\
    p_1, p_2, \cdots, p_n &= \arg\max_{p_i \in \mathbb{R}^+} \Tilde{U}_P^{t^*}(\Delta).
\end{align}
\end{definition}
Finding this equilibrium can seem a daunting problem at first, since the choice of the price menu affects the trader action in non linear ways, there is no guarantee that the action chosen by the trader would have the best utility for the LP. However, we now define the notion of an \textit{indifference menu}, which would turn out to be utility maximizing for the LP.

\subparagraph{The indifference menu for the traders} Out of the $n+1$ choices that the trader has, the trader chooses the one that maximizes his expected utility. This expected utility of the trader can be considered to be a good proxy for how strongly the trader prefers one action over another. Intuitively, the liquidity provider can imagine that there might be a way to make the trader indifferent between all choices of action, and then make her most favored action slightly more favorable to the trader, thereby obtaining a larger revenue. To formalize this intuition, we first find such an indifference menu of prices.
\begin{lemma}\label{lem:indiff}
    Let prices $p_i^*(\Delta)$ be given by 
    \begin{align}
        p_n^*(\Delta) &= \sup\{p>0:\expect{U_S((\sum_{i=1}^{n-1}Y_i + p)\Delta)} \leq \expect{U_S((\sum_{i=1}^{n}Y_i  )\Delta)}\},\\
        p_t^*(\Delta) &= p_{t+1}^*(\Delta) + \sup\{p>0:\expect{U_S((\sum_{i=1}^{t-1}Y_i + p + p_{t+1}^* )\Delta)} \leq \expect{U_S((\sum_{i=1}^{t}Y_i + p_{t+1}^* )\Delta)}\},
    \end{align}
    for all $t \in \{1,2,\cdots,n-1\}$.
    
    If the LP presents the menu $\{p_i^*(\Delta)\}_{i=1}^n$ to the trader, then the trader gets the same expected utility independent of each action. Furthermore, the $\{p_i^*(\Delta)\}_{i=1}^n$ constructed this way is the unique indifference menu.
\end{lemma}

\subparagraph{Indifference menu maximizes utility} Now, we show that any utility maximizing menu for the liquidity provider must also be an indifference menu for the trader.
\begin{lemma}\label{lem:20}
    Let $\{p_i^*\}_{i=1}^N$ be a utility maximizing menu for the LP. Then, $\{p_i^*\}_{i=1}^N$ is also an indifference menu for the trader.
\end{lemma}

The above lemma, combined with the uniqueness of the indifference menu, immediately give us the following result.
\begin{theorem}\label{thm:11}
    A menu of bonding curves $\{p_i^*\}_{i=1}^N$ is utility maximizing for the liquidity provider if and only if it is an indifference menu for the trader.
\end{theorem}


\subparagraph{A recipe for LPs} The results we have shown above give a simple procedure for a prospective liquidity provider to follow. Firstly, the LP brings knowledge of the risk aversion of the lenders/borrower who would be the primary traders the market wants to serve. Secondly, she expresses that knowledge using the constructive definition of an indifference menu as explicated in \prettyref{lem:indiff}. Thirdly, \prettyref{thm:11} guarantees that such an indifference menu will also maximize her expected utility over her wealth at maturity. She can then focus on tuning her estimates of how risk averse lenders and borrowers can be as market conditions change, and how traders are reacting to the menu. For instance, she should expect that an indifference menu garners roughly an equal number of trades for each of the $N$ possible actions, if her estimate of risk aversion is accurate. Any net deviation from this (for example, a large number of trades that choose one particular bonding curve in the menu) can be interpreted by her as an indication that she needs to change her estimates of risk aversion for the traders. Many existing data driven market making algorithms can be used for this purpose \cite{nadkarni2023zeroswap,nadkarni2024adaptive,cartea24strategic}.

\subparagraph{On the sustainability of liquidity provision} A caveat to the recipe for LPs we have outlined above is that the risk aversion of an LP must always be less than the risk aversion of traders. If this is not the case, even the \textit{optimally efficient} menu for the LP would quote worse prices (higher buying prices and lower selling prices) than the indifference menu for the trader. This implies that the best response for a trader would be to not buy/sell the yield token at all (in other words $t^* = n+1$ in \prettyref{eq:equil}). Thus, the hedging of lenders and borrowers as outlined in \prettyref{sec:welfare} can only take place efficiently if there are LPs that are less risk averse (or maybe even risk-neutral) than the lenders and borrowers.

\subparagraph{Aggregation across pools} Each liquidity provider may express her knowledge of the trader risk aversion differently, by using a different menu of bonding curves, constructed as per \prettyref{lem:indiff}. What remains to be done is to aggregate these expressions and provide the traders with a unified price quoting mechanism that eases price discovery. This can be done by converting all LP menu bonding curves into \textit{demand curves} and simply adding them up \cite{milionis2023myersonian}. This is the effective aggregation strategy used by many state-of-the-art trading interfaces like Tycho \cite{tycho}, Uniswap \cite{uniswap2023v4}, CoWSwap \cite{cowswap}, etc. A result of such aggregation would be that any trader can now view the current price of a yield token and the liquidity around the price (similar to viewing a limit order book) for a menu of buy/sell times into the future, till the maturity of the yield token. This would aid in price discovery of the yield token for lenders and borrowers to hedge their risk, while also assuring them that they are getting the best possible price from the available liquidity. This further enables the creation of fixed rate lending infrastructure, that we outline in the following section.

\subparagraph{Further liquidity aggregation via yield futures} So far, we have focused on market making for yield tokens of a specific maturity. However, as we saw in \prettyref{sec:model}, yield tokens derived from the same lending pools but different maturities are heavily correlated. The question then arises whether one can design an AMM to harness this correlation to further boost the liquidity of yield markets. Splitting a yield token into its constituent yield futures provides us with one such design. Suppose that any yield token \yieldToken, with maturity $T$, corresponding to a lending pool can be split into its constituent yield futures $\yieldToken_1,\cdots,\yieldToken_n$, each of them representing an interest rate payment for a particular block of the underlying blockchain. One can now use existing designs such as concentrated liquidity \cite{uniswap2021v3} to provide a menu of bonding curves, each of which corresponds to a particular yield future. A liquidity provider with yield tokens with maturity $T_1$ can then split her yield token into constituent $n_1$ futures and provide liquidity to all the $n_1$ bonding curves till maturity. Another LP with a yield token of maturity $T_2 (> T_1)$ can provide liquidity to $n_2$ bonding curves, with an overlap of liquidity in the first $n_1$ bonding curves. Yet another LP may only want to provide liquidity from time $T_1$ to $T_2$, and thus interact only with the intermediate $n_2-n_1$ bonding curves. Thus, any trading of yield tokens or futures would automatically use the liquidity of yield tokens of \textit{all maturities} that overlap in the price space, and also overlap temporally. We now formalize this intuition. First, we see how to convert the utility maximizing menu of bonding curves for yield \textit{tokens} into the same for yield \textit{futures}.
\begin{corollary}\label{cor:2}
    The utility maximizing menu of bonding curves for the yield futures $\yieldToken_1, \cdots, \yieldToken_n$ corresponding to a yield token \yieldToken is given by $\{P^*_i(\Delta)\}_{i=1}^n$, where $P^*_i(\Delta) = p^*_i(\Delta) - p^*_{i+1}(\Delta)$ and $\{p^*_i(\Delta)\}_{i=1}^n$ is the indifference menu of yield token bonding curves for the traders.
\end{corollary}


\section{Fixed interest rate lending}\label{sec:fixed}

In the previous sections, we saw the benefits yield tokens offer lenders and borrowers in managing their risk. We saw a trader's perspective on the yield token - the fair price, and also a liquidity provider's perspective on the yield token - a menu of bonding curves. In this section, we assume that we have lending pools that mints yield-bearing lending tokens for lenders, a yield tokenizer, and a reasonably liquid yield future AMM made by potentially distinct liquidity providers. We claim that these components are enough to build a fixed interest rate lending protocol.

\subparagraph{Doubly concentrated liquidity AMM} Each bonding curve in the menu presented in \prettyref{cor:2} can now be closely approximated by a concentrated liquidity position \cite{uniswap2021v3} in a canonical Uniswap-v3 style market maker \cite{robinson2021uniswap}. To see this, note that the liquidity in an infinitesimally narrow price range $[p,p+dp]$ of the curve corresponding to the $t^{th}$ yield future can be approximated by the expression $\frac{dx}{d(\sqrt{p})}$ \cite{robinson2021uniswap}, where $x$ is the amount of numeraire in the AMM reserve that obeys the LP's bonding curve $\psi$. One can then convert all bonding curves in \prettyref{cor:2} to a liquidity distribution over price ranges (Detailed construction in \prettyref{app:amm}). This liquidity distribution is additive across all LPs, thus providing a unified (hence, more liquid) market for not only yield tokens of the same maturity and different price preferences, but also yield tokens of different maturities. As in the canonical Uniswap v3 liquidity provision, the LP decides which price ranges (equivalent to an implied yield range as we saw in \prettyref{sec:pricing}) to provide liquidity. After that, the amount of yield tokens and the price range determines the liquidity level of the bonding curves and hence the amount of numeraire required to be deposited.

\subparagraph{Fixed term, fixed rate quoting} Suppose a lender queries the protocol for a fixed term loan of an asset amount $\Delta$ from time $T_1$ to $T_2$. The protocol then simulates the selling of yield futures in the AMMs corresponding to the asset, obtaining a price of selling the futures as $p_{\Delta}$ and quote a per-block interest rate  $p_{\Delta}/(n_2-n_1)$, where $n_1,n_2$ are the block numbers corresponding to $T_1,T_2$ respectively. On the other hand, if a borrower queries the protocol for a fixed term loan of an asset amount $\Delta$ from time $T_1$ to $T_2$, the protocol simulates a buy transaction on the AMM and quotes the corresponding interest rate. 

\subparagraph{Bringing liquidity and money markets on-chain} Once the borrower (lender) agrees on a quote, the protocol lends (deposits) the amount of asset from (to) the underlying lending pool and buys (sells) yield futures of the same amount for the term of the loan. Note that the yield futures markets can be minted on top of yield bearing tokens on individual lending pools or lending vaults (a mixture of many lending pools). This effectively builds a money market on top of existing lending vaults or pools such as Aave, Compound, Morpho, etc. We conjecture that such a market would bring \textit{more liquidity} and \textit{interest rate discovery} on chain, since the existing protocols more sophisticated users who can effectively manage variable interest rate risk. This is because a fixed term interest rate quoting protocol creates a public good in the form of a stable, aggregated interest rate. Borrowers who want access to capital to use elsewhere, and lenders who have access to capital can thus get the best interest rate possible, without having to know the inner workings of esoteric lending pools. 

\subparagraph{Lending aggregation} Apart from bringing more liquidity into lending pools, using a yield token AMM acts as a lending liquidity aggregator by harnessing the DEX aggregation capabilities that already exist. Thus, loans of every maturity can be viewed as an orderbook \cite{tycho} with a specific market price. Comparing these interest rates across varying collateral, any user would be able to lend/borrow tokens at the best possible interest rate, with the depth of the orderbook serving as a proxy for how much the market  is certain about the current interest rates. We give explicit details of the procedure outlined above to convert the bonding curves from \prettyref{cor:2} to an aggregated order book in \prettyref{app:amm}.


%% file: conclusion.tex
In this work, we have formalized the role of yield tokenization as a foundational primitive in decentralized finance, enabling more precise risk transfer and price discovery in the presence of volatile yield-generating activities. By modeling agents with heterogeneous beliefs and designing market makers that aggregate liquidity across diverse risk preferences, we demonstrate how programmable tokenization can replicate and extend core financial infrastructure within a permissionless setting. Our proposed mechanism improves upon existing protocols by allowing liquidity providers to specify custom bonding curves and by mitigating liquidity fragmentation across maturities through the use of yield future tokens. 

\subparagraph{Empirical extensions} Looking forward, future work may explore empirical fitting of these models from data on-chain, and integrate more adaptive learning dynamics between liquidity providers, speculators, and protocol curators for greater robustness and decentralization.

\subparagraph{Theoretical extensions} Tokenizing yield for LP tokens, liquid staking tokens also opens up hedging opportunities for users. In the case of liquid staking tokens especially, yield tokenization can lead to gas fee hedging and slashing risk insurance. We provide some preliminary results in \prettyref{app:lsd}, with the full design of slashing insurance as an open question.  On the flip side, a liquid yield market for staking may exacerbate the principal-agent problems in liquid staking/restaking (See for e.g. \cite{tzinas2023principal} and \prettyref{app:lsd}). We leave the fair pricing and quantifying the systemic risk and/or consensus related externalities of other such yield tokens to future work.


%% file: acknowledgements.tex
\section{Acknowledgements}
The authors thank a generous gift from XinFin Private Ltd for supporting this research. We were also supported by Army Research Office (ARO) Award W911NF2310147 and NSF SaTC 2325478. We also thank Xuechao Wang, Tarun Chitra, Ciamac Moallemi, Aviv Zohar, Nico Pei for fruitful comments and discussion.

%% file: appendix.tex
\section{Proof of \prettyref{thm:1}}
\begin{proof}
By Ito's formula, we get
\begin{align}\label{eq:ito1}
    d\priceYT = \left(\frac{\partial \priceYT}{\partial t} + \bar{\mu_t} \cdot \nabla \priceYT + \frac{1}{2}  \nabla \cdot \left(\Sigma_t \Sigma_t^T\right) \nabla \priceYT\right) dt + \Sigma_t d\bar{W_t} \cdot \nabla \priceYT
\end{align}
Suppose now, that we have a portfolio consisting of one unit of the yield token \yieldToken and $-\frac{\partial \priceYT}{\partial X^i_t}$ amounts of the token $\yieldBearingToken_i$, for all underlying tokens $\{\yieldBearingToken_i\}$. Let the value of the portfolio in terms of the numeraire \numeraire be $V_t$. After a time $dt$, the change in the portfolio value is given by
\begin{align}
    dV_t &= d\priceYT - \sum_i \frac{\partial \priceYT}{\partial X^i_t} dX_t^i + \sum_i Y^i(t,\bar{X_t}) X_t^i dt\\
    &= d\priceYT - \nabla \priceYT \cdot d\bar{X_t} + \bar{Y}(t,\bar{X_t})\cdot \bar{X_t} dt \label{eq:portfolio1}
\end{align}
Substituting \prettyref{eq:ito1} in \prettyref{eq:portfolio1} gives us
\begin{align}
    dV_t &= \left(\frac{\partial \priceYT}{\partial t} + \bar{\mu_t} \cdot \nabla \priceYT + \frac{1}{2}  \nabla \cdot \left(\Sigma_t \Sigma_t^T\right) \nabla \priceYT\right) dt + \Sigma_t d\bar{W_t} \cdot \nabla \priceYT \\
    &- \nabla \priceYT \cdot \left(\bar{\mu_t}(\bar{X_t}) dt + \Sigma_t(\bar{X_t}) d\bar{W_t}\right) + \bar{Y}(t,\bar{X_t})\cdot \bar{X_t} dt\\
    &= \left(\frac{\partial \priceYT}{\partial t} + \bar{Y}(t,\bar{X_t})\cdot \bar{X_t} + \frac{1}{2}  \nabla \cdot \left(\Sigma_t \Sigma_t^T\right) \nabla \priceYT\right) dt \label{eq:lhs1}
\end{align}
Also, we know that the value of the investment at time $t$ is given by
\begin{align}
    V_t &= \priceYT - \nabla \priceYT \cdot \bar{X_t}
\end{align}
Investing this in the risk free numeraire denominated investment with continuously accrued rate $r(t)$, we get the gains $dV'_t$given by
\begin{align}\label{eq:riskfreegains}
    dV'_t = \left(\priceYT - \nabla \priceYT \cdot \bar{X_t}\right) r(t) dt
\end{align}
Using \prettyref{eq:riskfreegains} and \prettyref{eq:lhs1}, we get the following partial differential equation from the no-arbitrage condition ($dV_t = dV'_t$) that we outlined as our notion of ``fair pricing''.
\begin{align}
    \frac{\partial \priceYT}{\partial t} + \frac{1}{2}  \nabla \cdot \left(\Sigma_t \Sigma_t^T\right) \nabla \priceYT + \bar{X_t} \cdot \left(\bar{Y}(t,\bar{X_t}) + r(t)\nabla \priceYT\right) - r(t) \priceYT = 0 \label{eq:pde_main}
\end{align}
where we have $\priceYT(T,x) = 0$ as a boundary condition.

To solve this PDE, we first define
\begin{align}
    Z(s) &:= e^{-\int_t^s r(u) du} \priceYT(s,\bar{X_s})
\end{align}
By Girsanov's Theorem (see Theorem 8.6.4 in \cite{oksendal2003stochastic}), we know that there exists a change of probability measure that transforms the Brownian motion $W_t$ to some $W^*_t$ so that we can write
\begin{align}
    d\bar{X_t} = r(t)\bar{X_t} dt + \Sigma_t d\bar{W_t^*}
\end{align}

Applying Ito's rule to evaluate $dZ(s)$ using the above expression, we get
\begin{align}
    dZ_s &= e^{-\int_t^s r(u) du} \left( -r(s)\priceYT + \frac{\partial \priceYT}{\partial s} + \frac{1}{2}  \nabla \cdot \left(\Sigma_s \Sigma_s^T\right) \nabla \priceYT + \bar{X_s} \cdot r(t)\nabla \priceYT \right) dt \\ &+ e^{-\int_t^s r(u) du} \Sigma_s d\bar{W}_s^*
\end{align}
Taking expectations conditioned on $\bar{X_t} = x$ on both sides, using \prettyref{eq:pde_main} and integrating $s$ from $t$ to $T$, using the fact that the price of the yield token is $0$ at maturity, gives us the result.
\end{proof}

\section{Proof of \prettyref{cor:1}}
\begin{proof}
    We prove that there is no arbitrage opportunity between a pair of maturities. The general result follows directly from that, since any multi-trade arbitrage would be a composition of such pairwise trades. 
    Suppose that one has yield tokens of maturities $T_1, T_2$. The prices must obey Ito's Rule, as given by
    \begin{align}
        dP_{\yieldToken}^{T_i} = \left( \frac{\partial P_{\yieldToken}^{T_i}}{\partial t} + \frac{1}{2}\sigma^2 X^2 \frac{\partial^2 P_{\yieldToken}^{T_i}}{\partial X^2} + X(\mu - Y(t,X))\frac{\partial P_{\yieldToken}^{T_i}}{\partial X}\right) dt + \sigma X \frac{\partial P_{\yieldToken}^{T_i}}{\partial X} dW
    \end{align}
    for $i = 1,2$.

    We can now form a portfolio consisting of $\frac{\partial P_{\yieldToken}^{T_2}}{\partial X}$ amount of the yield token with maturity $T_1$ and $-\frac{\partial P_{\yieldToken}^{T_1}}{\partial X}$ amount of the yield token with maturity $T_2$. Doing this cancels out the $dW$ and gives only the risk free drift term $dt$. The coefficient of $dt$ is nothing but the PDE derived above (\prettyref{eq:pde_main}), of which the fair prices in \prettyref{thm:1} are a solution. This implies that the drift term must be zero. This implies that there is no risk free arbitrage opportunity between a pair of yield tokens with different maturities.
\end{proof}

\section{Proof of \prettyref{lem:1}}
\begin{proof}
If the proposed equation does not hold, we have the following cases.
\begin{itemize}
    \item \textbf{Case I: } $(\priceYT > E[Y])$ In this case a lender can open an additional lending position in the underlying pool and hence mint a lending token $\yieldBearingToken$, use the tokenizer to mint a corresponding yield token $\yieldToken$ and sell it on the market to earn a expected profit of $\priceYT - E[Y]$.
    \item \textbf{Case II: } $(\priceYT < E[Y])$ In this case a borrower can buy $\yieldToken$ on the market and hold it till maturity to earn a expected profit of $E[Y] - \priceYT$.
\end{itemize}
\end{proof}

\section{Proof of \prettyref{lem:4}}
\begin{proof}
We have the following two cases.
\begin{itemize}
    \item \textbf{Case I (Lender):} Let \lender lend an amount $L$ of the asset to a lending pool with utilization $0< u \leq  1$. If \lender decides not to tokenize her position, then the expected utility over her wealth at maturity is $E[\utilLender(L + LY)]$. On the other hand, if she decides to tokenize her position and sell it at a fair price, her utility at maturity would be $\utilLender(L+LE[Y])$. Since \utilLender is concave and increasing, by Jensen's inequality, there exists a $\deltaLender \geq 0$ such that $\utilLender(L+L(E[Y]-\deltaLender)) = E[\utilLender(L+LY)]$.
    \item \textbf{Case II (Borrower):} Let \borrower borrow an amount $B = uL$ of the asset from a lending pool with utilization $0< u \leq  1$. If \borrower decides not to buy yield tokens to hedge his position, then the expected utility over his wealth at maturity is $E[\utilBorrower(-B - B \times Y/u)]$. On the other hand, if he decides to hedge his position by buying yield tokens at a fair price, he would have to buy an amount of yield tokens that would \textit{pay} him $BY/u$. To do this, he has to buy $B/u$ amount of yield tokens. Thus, his utility at maturity would be $\utilBorrower(-B-B/u \times E[Y])$. Since \utilBorrower is concave and increasing, by Jensen's inequality, there exists a $\deltaBorrower \geq 0$ such that $\utilBorrower(-B-B/u(E[Y]+\deltaBorrower)) = E[\utilBorrower(-B-BY/u)]$.
\end{itemize}    
Taking \deltaBorrower,\deltaLender as defined above, by Jensen's inequality, we know that, for any $\priceYT \in [E[Y] - \deltaLender,E[Y] + \deltaBorrower]$, we have $\utilBorrower(-B-B(\priceYT) \geq E[\utilBorrower(-B-BY/u)]$ and $\utilLender(W+W(\priceYT) \geq E[\utilLender(L+LY)]$. Thus, we see that a lender prefers to sell her yield tokens, while borrower prefers to buy yield tokens to hedge the interest rate payment on his loan.
\end{proof}

\section{Proof of \prettyref{thm:5}}
\begin{proof}

\begin{enumerate}

    \item Assume that the lending pool has a utilization $u \in (0,1]$. Then, if the total amount of the lent asset in the pool is $L=\sum L_i$, the total amount being borrowed is $uL = \sum B_i$. From the proof of \prettyref{lem:4}, we see that, at the fair price $E[Y]$, all lenders (no matter what level of risk aversion) are willing to sell their yield tokens and all borrowers (no matter what level of risk aversion) are willing to buy yield tokens. The number of yield tokens created by the lenders would be $LY/Y = L$, while the number of yield tokens required by borrowers to completely hedge their positions is also $(\sum B_i)Y/u = L$. This proves the first statement.
    \item Without a yield tokenizer, the welfare of the lending pool participants would be $W = \sum_{\lender_i}E[U_{\lender_i}(L_i + L_i Y)] + \sum_{\borrower_j}E[U_{\borrower_j}(B_j + B_j Y/u)] $. In presence of a yield tokenizer, the lenders and borrowers would prefer to hedge and hence their welfare becomes $W' = \sum_{\lender_i}U_{\lender_i}(L_i + L_i E[Y]) + \sum_{\borrower_j}U_{\borrower_j}(-B_j - B_j/u\times E[Y]) $. By Jensen's inequality, we know that $W \leq W'$.
\end{enumerate}
    
\end{proof}

\section{Proof of \prettyref{lem:8}}
\begin{proof}
    Suppose the LP starts with a reserve amounts of $x,y$ for the numeraire and yield token respectively. If she decides to buy an amount $\Delta$ of the yield token for a price $p$, her wealth at maturity is $x - p\Delta  + Y (y+\Delta)$. Let $x+Yy = Z$. Then, the utility of the LP in buying token amounts $\Delta_1,\Delta_2$ are respectively $U_P(Z + \Delta_1(Y-p_1)),U_P(Z + \Delta_2(Y-p_2))$. By Jensen's inequality, and the fact that $U_P$ is increasing concave, we know that
    $E[U_P(Z)] = E[E[U_P(Z)|Z]] \geq E[U_P(Z+\Delta_i(Y-E[Y]))]$. Therefore, there exists some $p_i \leq E[Y]$ such that $E[U_P(Z)] = E[U_P(Z+\Delta_i(Y-p_i))]$. This implies that 
    \begin{align}
        E[U_P(Z+\Delta_1(Y-p_1))] = E[U_P(Z+\Delta_2(Y-p_2))]\label{eq:9}
    \end{align} Setting $Z' = Z + \Delta_1(Y-p_1)$, we know that 
    \begin{align}
        E[U_P(Z')] &\geq  E[U_P(Z'+(\Delta_2-\Delta_1)(Y-E[Y]))] \\ &= E[U_P(Z+\Delta_2 Y-\Delta_2 E[Y]+\Delta_1 (E[Y]-p_1)))] \\ 
        &=  E[U_P(Z+\Delta_2 (Y- E[Y]+\Delta_1/\Delta_2 (E[Y]-p_1))))] \label{eq:10}.
    \end{align}
    
    Let us now focus on the term $E[Y] - \Delta_1/\Delta_2 (E[Y]-p_1)$. We can write it as $(1-\Delta_1/\Delta_2) E[Y] + \Delta_1/\Delta_2 p_1$. This term is always $\geq p_1$. Thus, we can write \prettyref{eq:10} as 
    \begin{align}
        E[U_P(Z+\Delta_1(Y-p_1))] \geq E[U_P(Z+\Delta_2(Y-p_1))]
    \end{align}
    Comparing this with \prettyref{eq:9} implies that $p_1 > p_2$.
\end{proof}

\section{Proof of \prettyref{thm:10}}
\begin{proof}
From the proof of \prettyref{lem:8}, we see that $p_\Delta^S$ is exactly how we defined $p_i$ corresponding to a $\Delta_i$. This implies that the liquidity provider faces the same expected utility over her wealth at maturity no matter where her reserves move along the curve $\psi_S$ as defined. A similar argument holds for the buy side of the curve using \prettyref{lem:9}. 

By the above lemmas, we also see that $p^S(\Delta)$ and  $p^B(\Delta)$ are non-increasing and non-decreasing functions of $\Delta$ respectively. Therefore, the curve $\psi$ that is parametrized by $\Delta$ is convex and non-increasing when written as a function of $x$ in the $xy-$space. This implies that this curve is a well-defined, incentive compatible AMM bonding curve \cite{Angeris_2020}.
\end{proof}

\section{Proof of \prettyref{lem:indiff}}
\begin{proof}
    We prove the above claim by induction. For the base case, we show that the trader is indifferent between selling the yield token at time $t=n$ and not selling the yield token at all. 
    To start with, we know that the trader would prefer selling the yield token over not selling it if
    \begin{align}
        \Tilde{U}_S^{n}(\Delta)&>\Tilde{U}_S^{n+1}(\Delta)\\
        \iff \expect{U_S((\sum_{i=1}^{n-1}Y_i + p_n)\Delta)} &> \expect{U_S((\sum_{i=1}^n Y_i)\Delta)}\label{eq:123}
    \end{align}
    Let us now investigate when \prettyref{eq:123} holds. By Jensen's inequality, we know that
    \begin{align}
        U_S((\sum_{i=1}^{n-1} Y_i + \expect{Y_n|Y_1,\cdots,Y_{n-1}})\Delta) &\geq \expect{U_S((\sum_{i=1}^nY_i)\Delta)|Y_1,\cdots,Y_{n-1}}.
    \end{align}
    This implies that $p_n^*(\Delta)$ is well defined. Taking $p_n = p_n^*(\Delta)$ as defined above, and taking expectations over $Y_1,\cdots,Y_n$ on both sides, we get
    \begin{align}
        \expect{U_S((\sum_{i=1}^{n-1} Y_i + p_n^*(\Delta))\Delta)} &= \expect{U_S((\sum_{i=1}^nY_i)\Delta)}
    \end{align}
    which shows that the trader is indifferent between selling the yield token at time $n$ and not selling the yield token.
    
    Assuming that this condition holds for any $t+1 \leq n$, we can similarly state that
    \begin{align}
        U_S((\sum_{i=1}^{t-1} Y_i + \expect{Y_t|Y_1,\cdots,Y_{t-1}} + p_{t+1}^*)\Delta) &\geq \expect{U_S((\sum_{i=1}^t Y_i + p_{t+1}^*)\Delta)|Y_1,\cdots,Y_{t-1}}.
    \end{align}
    This implies that $p_t^*(\Delta)$ is well defined and hence 
    \begin{align}
        \expect{U_S((\sum_{i=1}^{t-1} Y_i + p_t^*(\Delta) )\Delta)} &= \expect{U_S((\sum_{i=1}^nY_i + p_{t+1}^*)\Delta)},
    \end{align}
    which implies that the trader is indifferent between the selling at time $t$ and at time $t+1$.

    By induction, we see that, for the menu of prices defined above, the trader achieves the same expected utility over his wealth at maturity no matter when he chooses to sell the yield token.

    Finally, we know that any indifference menu has to satisfy 
    \begin{align}
        \expect{U_S((\sum_{i=1}^{t-1} Y_i + p_t^*(\Delta) )\Delta)} &= \expect{U_S((\sum_{i=1}^{n} Y_i )\Delta)}
    \end{align}
    by definition. Since $U_S$ is an increasing function, this equation can be satisfied by a single unique value of $p_t^*(\Delta)$.
\end{proof}

\section{Proof of \prettyref{lem:20}}
\begin{proof}
    Suppose that $\{p_i^*\}_{i=1}^N$ is a utility maximizing menu but is not an indifference menu. This implies that there is an optimal action $t^*$ for the trader, and that there for all $t' \neq t^*$, there exists $\epsilon_{t'}>0$ such that $\Tilde{U}_S^{t^*} = \Tilde{U}^{t'}_S + \epsilon_{t'} $. Let $t = \arg\min_{t'} \epsilon_{t'}$.
    
    One can now define a new menu $\{p'_i\}_{i=1}^N$ such that $p'_i = p^*_i$ for $i \neq t^*$ and $p'_{t^*} = p^*_{t^*} - \delta$, where $\delta>0$ is still small enough so that the optimal action of the trader is to sell the yield token at time $t^*$ over selling at $t$. However, for this new menu, we will have the utility of the LP to be $U_P':=\expect{U_P(x-y(\sum_{i=1}^n Y_i) - \Delta (p'_{t^*}) + \Delta (\sum_{i=t^*}^n Y_i))}$. Since $U_P$ is an increasing function, we have that the value $U_P'$ must be greater than the value $U_P^* := \expect{U_P(x-y(\sum_{i=1}^n Y_i) - \Delta (p^*_{t^*}) + \Delta (\sum_{i=t^*}^n Y_i))}$ that the LP gets when presenting the menu $\{p_i^*\}_{i=1}^N$. This contradicts our assumption the $\{p_i^*\}_{i=1}^N$ is utility maximizing for the LP.
\end{proof}

\section{Yield AMM aggregation}\label{app:amm}

\input{aggregation}

\section{Yield tokenization in liquid staking}\label{app:lsd}

\input{restaking}

%% file: aggregation.tex
In this section, we collate some canonical results and knowledge about CFMMs and their properties.

\subparagraph{Concentrated liquidity} Concentrated liquidity, as introduced in Uniswap v3, allows liquidity providers (LPs) to allocate their capital within specific price ranges rather than across the entire price spectrum. This targeted approach enhances capital efficiency and enables LPs to earn higher fees with less capital by focusing liquidity where trading activity is most intense. In traditional constant mean market AMMs, the invariant is defined as:
\[
x^\theta \cdot y^{1-\theta} = L
\]
where \( x \) and \( y \) are the reserves of two assets, and \( L \) is a constant. This model distributes liquidity uniformly across all price levels.

In contrast, Uniswap v3's concentrated liquidity model allows LPs to specify a price range \([p_l, p_u]\) within which they provide liquidity. The amount of liquidity \( L \) provided within this range determines the virtual reserves of the assets at any price \( p \) within the range:
\[
x(p) = 2L(1-\theta)(p^{\theta-1} - p_u^{\theta-1}), \quad y(p) = 2L \theta(p^\theta - p_l^{\theta})
\]
These equations ensure that the bonding curve is obeyed within the specified range. By allowing liquidity to be concentrated within specific price intervals, Uniswap v3 enables the approximation of various static AMM curves. For instance, by strategically placing multiple liquidity positions across different price ranges, one can emulate the liquidity distribution of AMMs like Curve or Balancer. This flexibility allows for the replication of desired price curves and trading behaviors within the Uniswap v3 framework. Before we do that explicitly, we define the fundamental concept of a \textit{demand curve}. 

\subparagraph{The many avatars of AMMs} A demand curve represents the minimum amount of an asset the liquidity provider wants to hold at a given price. Thus, any bonding curve $\psi(x,y) = L$ can be converted into a demand curve by simply expressing the reserves of the asset as a function of price $x(p)$. The reserves of the numeraire are implied by the demand since they are tied with each other via the price of the asset. That is, we have $y(p) = -\int_\infty^p \pi x'(\pi) d\pi$. This representation of the AMM now easily allows us to split and aggregate liquidity across multiple liquidity providers. The demand curves for each LP can be simply added to give us the demand curve of the market as a whole. Another representation that gives us a much more intuitive handle on the notion of where liquidity is present, is the derivative of the demand curve $L(p)=-x'(p)$, which we call the \textit{liquidity curve}. This is a measure of how much the liquidity provider is willing to sell when the marginal price of the asset changes from $p$ to $p+dp$. This is the ``orderbook'' view of the market that is also additive across liquidity providers. To summarize, the AMM has three equivalent representations - the bonding curve \cite{Angeris_2020} (constrains the reserves of the AMM directly in $x,y-$space), the demand curve \cite{milionis2023myersonian}(represents the minimum amount of asset that the AMM wants to hold as inventory at a price $p$), and the liquidity curve (represents the infinitesimal/marginal amount of asset that the AMM wants to sell at a certain price $p$). For the AMM to be incentive compatible (which is necessary for the AMM to be non-trivial and a valid price oracle), the bonding curve function has to be quasiconcave \cite{Angeris_2020}. This implies that the demand curve has to non increasing \cite{milionis2023myersonian} and hence the liquidity curve has to be non-negative. Continuing our concentrated liquidity example from above, we see that the liquidity curve corresponding to a basic Uniswap v3 position is $L(p) = \frac{L}{2p^{3/2}}\mathbbm{1}_{p\in [p_l,p_u]}$.

\subparagraph{Approximations} Any liquidity curve can now be approximated by a Uniswap-v3 style market maker. To see this, note that any liquidity curve is a sum of constituent liquidity curves that each only have liquidity in mutually exclusive intervals on the number line representing prices. In other words, if $L(p)$ is any non-negative function of the price $p$, we can write it as $L(p) = \sum_i L_i \mathbbm{1}_{p\in [p_i,p_{i+1}]}$, where $p_i = i \delta$ for a small $\delta$. Now, each basis function $L_i \mathbbm{1}_{p\in [p_i,p_{i+1}]}$ can be approximated as $2p_i^{3/2}\frac{L_i}{2p^{3/2}}\mathbbm{1}_{p\in [p_i,p_{i+1}]}$. Thus, for a $\delta$ small enough, we can represent \textit{any} liquidity curve (and hence any bonding curve $\psi$) as a linear combination of concentrated liquidity positions as $L(p) \approx \sum_i 2L_ip_i^{3/2} l_{p_i,p_{i+1}}$, where $l_{p_l,p_u} := \frac{1}{2p^{3/2}}\mathbbm{1}_{p\in [p_l,p_u]}$ is our basis concentrated liquidity position.

\subparagraph{Putting it all together} Thus, any LP, or a group of LPs can aggregate their liquidity of yield futures across the menu of bonding curves by converting all bonding curves into liquidity curves and simply summing them up. Also, the distribution of trading fees for any trade would simply be the pro-rata share of the value of a liquidity curve of a particular liquidity provider to the sum at any price point.

%% file: restaking.tex
\subparagraph*{Staking, restaking, liquid staking} \textit{Staking} is an action used by native token holders on a Proof-of-Stake blockchain to secure the underlying consensus and execution mechanisms. By default, staking involves locking in one's wealth in a native token (such as ETH), and participating in block attestation, validation, and proposal \cite{lido2025}. The locked-in stake earns rewards in the native token on successful execution of the protocol, but may be \textit{slashed} in the event of validator misbehavior as defined by the requirements of the underlying consensus mechanism \cite{lidoslashing2023}. Recently, the rise of \textit{restaking} via platforms such as Eigenlayer \cite{eigenlayer2025}, Babylon \cite{babylon2025}, Mesh Security in Cosmos \cite{meshsecurity2023}, etc. allow validators to take on additional responsibilities in Actively Validated Services, while using the same underlying stake as a slashable guarantee that the service will be validated according to protocol \cite{eigenlayer2025, babylon2025, meshsecurity2023}. This has allowed many application specific protocols to leverage security guarantees and validator networks of existing blockchains like Ethereum and Bitcoin without having to build an application specific blockchain from scratch. However, because staking and restaking deprive a validator's stake from participating in other, potential yield-earning activities such as lending, \textit{liquid staking} has been introduced as an alternative. In a liquid staking protocol (such as LIDO \cite{lido2025}, RocketPool \cite{rocketpool2025}), a user deposits a native token (say \textit{ETH}) and delegates validation responsibilities to a validator, or set of validators. In exchange, the user is minted liquid staking tokens (such as \textit{staked ETH} or \textit{stETH}) that act as a yield bearing token. The yield that a liquid staking protocol earns is in the form of fixed block proposal and attestation rewards, variable transaction and blob priority fees, sync committee participation, etc.

\subparagraph{Model of liquid staking} In this section, we assume a simple model of liquid staking/restaking. First, we have a set of token holders of the native token \numeraire, that we also assume to be our numeraire. These token holders delegate their stake to a liquid staking protocol, that we represent by a single entity $\mathcal{D}$. This protocol then mints \textit{staked }\numeraire or \textit{st}\numeraire, represented as a yield bearing token $\yieldBearingToken$. The yield is denominated in the native token \numeraire directly, and is represented by the discrete time process $\{Y_t\}_{t=1}^n$. Each yield reward $Y_t > 0$ consists of a fixed part (from attestation and block rewards) and a variable part (from transaction fees). The yield tokenizer takes in \yieldBearingToken and mints a principal token \principalToken (representing the redeemable \yieldBearingToken at maturity) and yield token \yieldToken (representing the yield till maturity). In case the delegated validator misbehaves, we assume that stake of \numeraire get slashed by a factor $p < 1$. We assume this same model also holds for \textit{liquid restaking tokens}, but since LRTs may have a higher slashing risks associated with them \cite{tzinas2023principal}, we just assume that the value of $p$ is higher in that case.

\subsection{Fair price of risk-free staking yield}
Many liquid staking protocols such as LIDO maintain a whitelist of approved validators that are trusted delegates. In such an environment, the holders of liquid staking tokens may assume that their positions are free of slashing risk - indeed, incidents of slashing have been incredibly rare \cite{lidoslashing2023}. In this scenario, we examine how one might price principal and yield tokens. One might guess that the fair price of the yield token should simply be the sum of the expected values of the yields. However, one must also take into account the fact that the yield emitted by the yield token may be staked again to earn risk-free returns.

\begin{lemma}\prettyref{lem:lsd}
    Assume that the liquid staking token \yieldBearingToken has been locked in the yield tokenizer for only one time slot with yield at maturity $Y$. Then, the fair prices of the principal token and yield token are given by
    \begin{align}
        p_{\principalToken} &= \frac{1}{1+\expect{Y}}\\
        p_{\yieldToken} &= \frac{\expect{Y}}{1+\expect{Y}}
    \end{align}
\end{lemma}
\begin{proof}
    Firstly, since a principal token and yield token can be combined and burned to unlock the underlying liquid staking token, we have $p_{\principalToken} + p_{\yieldToken} = 1$. Focusing on the price of \principalToken, we have two cases if the given prices are not fair.
    \begin{itemize}
        \item \textbf{Case I: }$p_{\principalToken} (1+\expect{Y}) > 1.$ In this case, a holder of \principalToken is incentivized to sell for the underlying $p_{\principalToken}$ amount of \numeraire and stake it to obtain a higher yield than holding it.
        \item \textbf{Case II: }$p_{\principalToken} (1+\expect{Y}) < 1.$ In this case, an arbitrageur can simply buy a unit \principalToken for a price cheaper than a unit of \numeraire, and hold it to maturity to obtain a unit of \numeraire, thereby making a profit.
    \end{itemize}
    This implies that the fair price for the principal token must be as stated.
\end{proof}
This result can be further extended to multiple yield payments.
\begin{lemma}\label{lem:lsd}
    For the general liquid staking token \yieldBearingToken with multiple yield payments $\{Y_t\}_{t=1}^n$ till maturity, we have
    \begin{align}
        p_{\principalToken} &= \frac{1}{\expect{\prod_{i=1}^n (1+Y_i)}}\\
        p_{\yieldToken} &= 1-p_{\principalToken}
    \end{align}
\end{lemma}

\subsection{Externalities - Insurance}
\subparagraph{Hedging congestion} The variable part of the yield obtained from a yield token $\yieldToken$ on a liquid staking token is due to transaction and blob fees paid to a blockchain's validator. Therefore, in a blockchain such as Ethereum, where the validators of each block are declared every epoch, buying the corresponding yield tokens can act as means of hedging against transaction/blob congestion. On the sell side, they would also help in freezing the yields obtained from staking to a specific delegate - i.e. a native token holder can stake in a liquid staking protocol, tokenize her yield, and sell it to freeze the yield for a specific time period. Therefore, transaction fee payers act as buyers of the token for hedging congestion risk and liquid staking delegators act as sellers for hedging risk of staking rewards going down. Thus, the analysis we performed on the welfare improvement of lenders and borrowers for lending pools carries through (\prettyref{sec:welfare}). To make this risk transfer efficient, one can similarly design better AMMs depending on the risk aversion of these two types of participants (\prettyref{sec:amm}). 

\subparagraph{Slashing insurance} Even though slashing is a rare event, liquid stakers may way want to insure against such an event. In the case of liquid \textit{restaking}, this risk may be even higher \cite{neuder2024risks}. In such cases, we demonstrate how holding the principal token can act as a form of insurance. Suppose that a liquid staker uses a yield tokenizer to mint \principalToken and \yieldToken. Selling \yieldToken would freeze her yield for the time till maturity. However, holding on to the \principalToken helps in managing slashing risk. This is because in case a slashing event is triggered on the delegate, the value of the principal token would rise from $\frac{1}{\expect{\prod_{i=1}^n (1+Y_i)}}$ to $\frac{1}{1-p}$, where $p$ is the proportion of stake that is slashed. Here, we have assumed that in case of slashing, one unit of principal token still represents one unit of underlying native token. This further assumes that the yield tokenenizer is insuring the liquid staker against slashing, with the insurance premium being collected via fees that we ignore in the following calculation. Thus, depending on how likely a liquid staker thinks slashing might occur in a given time period, she can buy more principal tokens to sell in case of a slashing event. More formally, the following lemma tells us precisely how much additional native token \numeraire she should spend to insure a unit amount of \numeraire staked via liquid staking.
\begin{lemma}
    Suppose that the liquid staker has staked a unit amount of the native token \numeraire till a maturity $T$, and ascribes a small probability $p_S << 1$ that a slashing event will occur in this time. Furthermore the staker has hedged against future yield fluctuations by tokenizing her yield and selling the yield token. To buy insurance against slashing, she must spend an amount $p_I$ in \numeraire to buy principal tokens, where $p_I$ is given by
    \begin{align}
        p_I = \frac{p_S }{1-p + p_S\expect{\prod_{i=1}^n (1+Y_i)}}
    \end{align}
\end{lemma}
\begin{proof}
    Let $K$ denote the number of principal tokens bought by the liquid staker. Using \prettyref{lem:lsd}, the expected profit from purchasing the insurance is
    \begin{align}
        (1-p_S)\frac{-K}{\expect{\prod_{i=1}^n (1+Y_i)}} + p_S \left(K \left(\frac{1}{1-p}-\frac{1}{\expect{\prod_{i=1}^n (1+Y_i)}}\right) + \frac{1}{1-p}\right).
    \end{align}
    Here, the first term represents the case when no slashing occurs, in which the staker loses the insurance premium (equal to the price of the extra principal tokens that were bought). The second term represents the profit obtained from selling the principal tokens in case of slashing. The last term arise from selling the original principal token that the staker minted for free for depositing her liquid staking token in the yield tokenizer. For fairness, the expected loss to insurance must be $0$. Setting the above expression to $0$ and solving for $K$ gives us the desired result.
\end{proof}
Extending the above lemma to the case where the yield tokenizer charges a non trivial insurance premium with their own model of slashing risk would be an illuminating future direction.

\subparagraph{Risk to consensus - malicious adverse selection} While the positive externalities of tokenizing yield in liquid staking bear similarity to lending protocols (hedging against various types of risk), the act of staking is inherently important to the underlying blockchain consensus, in a way that a lending pool is not. In particular, the fact that the principal token can be used as insurance opens up the possibility of adverse selection. For instance, a delegate can buy principal tokens (the insurance) and equivocate/misbehave on purpose (insurance fraud), thereby earning a profit from misbehaving while participating in the underlying consensus mechanism. This would exacerbate the principal-agent problem between the stakers and the delegates \cite{tzinas2023principal}. In fact, a \textit{more liquid} market would make this adverse selection problem \textit{worse}, since it would make buying large amounts of the principal token much easier. A liquid market would also open up the possibilities of easier ways to short the underlying liquid staking token, thus making an attack like the one outlined above more profitable. A solution to this might be to have \textit{exempt delegations} \cite{tzinas2023principal}, which means that the delegate also has to lock in their own tokens and subject them to the same slashing risk. While this may make such an attack less likely, it decreases capital efficiency for an honest delegate, who might then start charging liquid stakers a higher fee for providing the service. One can imagine the market finding an equilibrium liquid staking fee and exempt delegation amount, but these parameters would have to be updated as the liquidity of yield tokenized markets keep changing.

%% file: references.bib
@article{gudgeon2020defi,
  title={DeFi Protocols for Loanable Funds: Interest Rates, Liquidity and Market Efficiency},
  author={Gudgeon, Lewis and Werner, Sam M. and Perez, Daniel and Knottenbelt, William J.},
  journal={arXiv preprint arXiv:2006.13922},
  year={2020},
  url={https://arxiv.org/abs/2006.13922}
}

@misc{eip4626,
  title={EIP-4626: Tokenized Vault Standard},
  author={Ethereum Improvement Proposals},
  year={2022},
  url={https://eips.ethereum.org/EIPS/eip-4626}
}

@misc{pendle_docs,
  title={Pendle Finance Documentation},
  author={Pendle Finance},
  year={2021},
  url={https://docs.pendle.finance/}
}

@misc{element_paper,
  title={The Element Protocol Construction Paper},
  author={Element Finance},
  year={2021},
  url={https://paper.element.fi/}
}

@misc{spectra_docs,
  title={Spectra Finance Documentation},
  author={Spectra Finance},
  year={2023},
  url={https://docs.spectra.finance/}
}

@misc{notional_blog,
  title={Notional Finance Blog},
  author={Notional Finance},
  year={2021},
  url={https://blog.notional.finance/}
}

@misc{investopedia_strips,
  title={Stripped Treasury Bonds},
  author={Investopedia},
  year={2021},
  url={https://www.investopedia.com/terms/s/strippedtreasurybond.asp}
}

@misc{investopedia_cmo,
  title={Collateralized Mortgage Obligation (CMO)},
  author={Investopedia},
  year={2021},
  url={https://www.investopedia.com/terms/c/cmo.asp}
}

@misc{investopedia_irs,
  title={Interest Rate Swap},
  author={Investopedia},
  year={2021},
  url={https://www.investopedia.com/terms/i/interestrateswap.asp}
}

@misc{fluid2023whitepaper,
  title={Fluid Finance Whitepaper},
  author={Fluid Finance},
  year={2023},
  url={https://fluid.finance/whitepaper.pdf}
}

@misc{pendle2021whitepaper,
  title={Pendle V2 AMM Whitepaper},
  author={Pendle Finance},
  year={2021},
  url={https://raw.githubusercontent.com/pendle-finance/pendle-v2-resources/main/whitepapers/V2_AMM.pdf}
}

@misc{element2021whitepaper,
  title={The Element Protocol Construction Paper},
  author={Element Finance},
  year={2021},
  url={https://paper.element.fi/}
}

@misc{spectra2023docs,
  title={Principal \& Yield Token | Spectra},
  author={Spectra Finance},
  year={2023},
  url={https://docs.spectra.finance/core-concepts/principal-and-yield-token}
}

@misc{notional2021whitepaper,
  title={Notional Finance Whitepaper},
  author={Notional Finance},
  year={2021},
  url={https://notional.finance/whitepaper.pdf}
}

@misc{yearn2021whitepaper,
  title={Yearn Finance Whitepaper},
  author={Yearn Finance},
  year={2021},
  url={https://yfi.management/assets/whitepaperv1.pdf}
}

@misc{yieldspace2020paper,
  title={YieldSpace: An Automated Liquidity Provider for Fixed Yield Tokens},
  author={Allan Niemerg and Dan Robinson and Lev Livnev},
  year={2020},
  url={https://yield.is/yieldspace.pdf}
}

@misc{replicating2021arxiv,
  title={Replicating Market Makers},
  author={Guillermo Angeris and Alex Evans and Tarun Chitra},
  year={2021},
  url={https://arxiv.org/abs/2103.14769}
}

@misc{nezlobin2025lossversusrebalancingdeterministicgeneralizedblocktimes,
      title={Loss-Versus-Rebalancing under Deterministic and Generalized block-times}, 
      author={Alex Nezlobin and Martin Tassy},
      year={2025},
      eprint={2505.05113},
      archivePrefix={arXiv},
      primaryClass={q-fin.MF},
      url={https://arxiv.org/abs/2505.05113}, 
}

@misc{agilerate2024arxiv,
      title={AgileRate: Bringing Adaptivity and Robustness to DeFi Lending Markets}, 
      author={Mahsa Bastankhah and Viraj Nadkarni and Xuechao Wang and Pramod Viswanath},
      year={2025},
      eprint={2410.13105},
      archivePrefix={arXiv},
      primaryClass={cs.SI},
      url={https://arxiv.org/abs/2410.13105}, 
}

@misc{marketdesign2025paper,
  title={A Unified Framework for Dynamic Prediction Market Design},
  author={Shipra Agrawal and Erick Delage and Mark Peters and Zizhuo Wang and Yinyu Ye},
  year={2025},
  url={https://web.stanford.edu/class/msande310/ORfinal.pdf}
}

@misc{mev2023springer,
  title={Routing MEV in Constant Function Market Makers},
  author={Kshitij Kulkarni and Theo Diamandis and Tarun Chitra},
  year={2023},
  url={https://link.springer.com/chapter/10.1007/978-3-031-48974-7_26}
}

@misc{tenor2024docs,
  title={Tenor Protocol Documentation},
  author={Tenor Finance},
  year={2024},
  url={https://www.docs.tenor.finance/overview/lend/}
}

@misc{kpmg2021crypto,
  title={Crypto Insights Part 2: Decentralised Exchanges and Automated Market Makers},
  author={KPMG},
  year={2021},
  url={https://assets.kpmg.com/content/dam/kpmg/cn/pdf/en/2021/10/crypto-insights-part-2-decentralised-exchanges-and-automated-market-makers.pdf}
}

@misc{radix2021amms,
  title={What are AMMs (Automated Market Maker) and CFMMs (Constant Function Market Maker)},
  author={Radix DLT},
  year={2021},
  url={https://learn.radixdlt.com/article/what-are-amms-automated-market-maker-and-cfmms-constant-function-market-maker}
}

@misc{linumlabs2019bonding,
  title={Bonding Curves - The What, Why, and Shapes Behind Them},
  author={Linum Labs},
  year={2019},
  url={https://www.linumlabs.com/articles/bonding-curves-the-what-why-and-shapes-behind-it}
}

@misc{uniswap2020v2,
  title={Uniswap v2 Mainnet Launch},
  author={Uniswap},
  year={2020},
  url={https://blog.uniswap.org/launch-uniswap-v2}
}

@article{chitra2025perpetual,
  title={Perpetual Demand Lending Pools},
  author={Chitra, Tarun and Diamandis, Theo and Sheng, Nathan and Sterle, Luke and Yusubov, Kamil},
  journal={arXiv preprint arXiv:2502.06028},
  year={2025},
  url={https://arxiv.org/abs/2502.06028}
}

@article{chitra2025curationary,
  title={A Curationary Tale: Logarithmic Regret in DeFi Lending via Dynamic Pricing},
  author={Chitra, Tarun},
  journal={arXiv preprint arXiv:2503.18237},
  year={2025},
  url={https://arxiv.org/abs/2503.18237}
}

@misc{moallemi2024pmamm,
  title={pm-AMM: A Uniform AMM for Prediction Markets},
  author={Moallemi, Ciamac and Robinson, Dan},
  year={2024},
  howpublished={\url{https://www.paradigm.xyz/2024/11/pm-amm}},
  note={Accessed: 2025-05-28}
}

@misc{balancer2021amm,
  title={Balancer AMM: What Is It, and How Does It Work?},
  author={Crypto News},
  year={2021},
  url={https://crypto.news/balancer-amm-what-is-it/}
}

@misc{uniswap2021v3,
  title={Concentrated Liquidity - Uniswap},
  author={Uniswap},
  year={2021},
  url={https://docs.uniswap.org/concepts/protocol/concentrated-liquidity}
}

@misc{uniswap2023v4,
  title={Hooks - Uniswap},
  author={Uniswap},
  year={2023},
  url={https://docs.uniswap.org/contracts/v4/concepts/hooks}
}

@article{barrons2025impermanent,
  title={Impermanent Loss Mitigation - Protecting Your DeFi Investments},
  author={Barrons Independent},
  year={2025},
  url={https://www.barrons-independent.com/impermanent-loss-mitigation-protecting-your-defi-investments/}
}

@article{coinflare2025liquid,
  title={Risks And Limitations Of Liquid Staking},
  author={CoinFlare},
  year={2025},
  url={https://www.coinflare.com/blog/risks-and-limitations-of-liquid-staking/}
}

@article{fireblocks2025liquid,
  title={Ethereum Staking \& Liquid Staking: Risks, Rewards \& Insights},
  author={Fireblocks},
  year={2025},
  url={https://www.fireblocks.com/report/liquid-staking-101/}
}

@inproceedings{werner2021sok,
  title={SoK: Decentralized Finance (DeFi)},
  author={Werner, Sam M and Perez, Daniel Perez and Gudgeon, Lewis and Klages-Mundt, Ariah and Knottenbelt, William and Gervais, Arthur},
  booktitle={IEEE European Symposium on Security and Privacy},
  year={2021}
}

@article{chen2020defiprim,
  title={DeFi and the Future of Finance},
  author={Chen, Jason and Masi, Collin and Wu, Meng and Brukhman, Jake},
  journal={ArXiv preprint arXiv:2102.08091},
  year={2020}
}

@article{schar2021defi,
  title={Decentralized Finance: On Blockchain- and Smart Contract-Based Financial Markets},
  author={Schär, Fabian},
  journal={Federal Reserve Bank of St. Louis Review},
  volume={103},
  number={2},
  pages={153--174},
  year={2021},
  publisher={Federal Reserve Bank of St. Louis},
  url={https://www.stlouisfed.org/publications/review/2021/02/05/decentralized-finance-on-blockchain-and-smart-contract-based-financial-markets}
}

@article{carapella2022defi,
  title={Decentralized Finance (DeFi): Transformative Potential and Associated Risks},
  author={Carapella, Francesca and Dumas, Edward and Gerszten, Jacob and Swem, Nathan and Wall, Larry},
  journal={Federal Reserve Board Finance and Economics Discussion Series},
  number={2022-14},
  year={2022},
  url={https://www.federalreserve.gov/econres/feds/decentralized-finance-defi-transformative-potential-and-associated-risks.htm}
}

@article{cfainstitute2022defi,
  title={DeFi-ing the Rules: Five Opportunities and Five Risks of Decentralized Finance},
  author={Harvey, Campbell R.},
  journal={CFA Institute},
  year={2022},
  url={https://blogs.cfainstitute.org/investor/2022/06/07/defi-ing-the-rules-five-opportunities-and-five-risks-of-decentralized-finance/}
}

@article{axios2022centrifuge,
  title={Centrifuge turns real world assets into loans},
  author={Axios},
  year={2022},
  url={https://www.axios.com/2022/09/23/centrifuge-turns-real-world-assets-into-loans}
}

@misc{chitra2023attacksdynamicdefirate,
      title={Attacks on Dynamic DeFi Interest Rate Curves}, 
      author={Tarun Chitra and Peteris Erins and Kshitij Kulkarni},
      year={2023},
      eprint={2307.13139},
      archivePrefix={arXiv},
      primaryClass={cs.CR},
      url={https://arxiv.org/abs/2307.13139}, 
}

@misc{robinson2021uniswap,
  title        = {Uniswap v3: The Universal AMM},
  author       = {Dan Robinson},
  year         = {2021},
  howpublished = {\url{https://www.paradigm.xyz/2021/06/uniswap-v3-the-universal-amm}},
  note         = {Accessed: 2025-05-27}
}

@article{cartea24strategic,
  title={Strategic Bonding Curves in Automated Market Makers},
  author={Cartea, Álvaro and Drissi, Fayçal and Sánchez-Betancourt, Leandro and Siska, David and Szpruch, Lukasz},
  url={https://ssrn.com/abstract=5018420},
  year={2024}
}

@InProceedings{bastankhah2024thinking,
  author =	{Bastankhah, Mahsa and Nadkarni, Viraj and Jin, Chi and Kulkarni, Sanjeev and Viswanath, Pramod},
  title =	{{Thinking Fast and Slow: Data-Driven Adaptive DeFi Borrow-Lending Protocol}},
  booktitle =	{6th Conference on Advances in Financial Technologies (AFT 2024)},
  pages =	{27:1--27:23},
  series =	{Leibniz International Proceedings in Informatics (LIPIcs)},
  ISBN =	{978-3-95977-345-4},
  ISSN =	{1868-8969},
  year =	{2024},
  volume =	{316},
  editor =	{B\"{o}hme, Rainer and Kiffer, Lucianna},
  publisher =	{Schloss Dagstuhl -- Leibniz-Zentrum f{\"u}r Informatik},
  address =	{Dagstuhl, Germany},
  URL =		{https://drops.dagstuhl.de/entities/document/10.4230/LIPIcs.AFT.2024.27},
  URN =		{urn:nbn:de:0030-drops-209634},
  doi =		{10.4230/LIPIcs.AFT.2024.27}
}

@article{vasicek1977equilibrium,
  title={An equilibrium characterization of the term structure},
  author={Vasicek, Oldrich},
  journal={Journal of Financial Economics},
  volume={5},
  number={2},
  pages={177--188},
  year={1977},
  publisher={Elsevier}
}

@article{cox1985theory,
  title={A theory of the term structure of interest rates},
  author={Cox, John C and Ingersoll, Jonathan E and Ross, Stephen A},
  journal={Econometrica: Journal of the Econometric Society},
  pages={385--407},
  year={1985}
}

@article{hull1990pricing,
  title={Pricing interest-rate-derivative securities},
  author={Hull, John and White, Alan},
  journal={The Review of Financial Studies},
  volume={3},
  number={4},
  pages={573--592},
  year={1990},
  publisher={Oxford University Press}
}

@article{heath1992bond,
  title={Bond pricing and the term structure of interest rates: A new methodology for contingent claims valuation},
  author={Heath, David and Jarrow, Robert and Morton, Andrew},
  journal={Econometrica},
  volume={60},
  number={1},
  pages={77--105},
  year={1992},
  publisher={Wiley}
}

@article{cohen2023economics,
  title={The Economics of Interest Rate Models in Decentralised Lending Protocols},
  author={Cohen, Samuel N. and S{\'a}nchez-Betancourt, Leandro and Szpruch, {\L}ukasz},
  journal={SSRN Electronic Journal},
  year={2023},
  url={https://ssrn.com/abstract=4638390},
  note={Accessed: 2025-05-27}
}

@article{cohen2023paradox,
  title={The Paradox of Adversarial Liquidation in Decentralised Lending},
  author={Cohen, Samuel N. and Sabate-Vidales, Marc and Szpruch, {\L}ukasz and Gontier Delaunay, Mathis},
  journal={SSRN Electronic Journal},
  year={2023},
  url={https://ssrn.com/abstract=4540333},
  note={Accessed: 2025-05-27}
}

@misc{compoundGov,
key={compoundGov},
  title = {Compound Finance Governance},
  howpublished = {\url{https://compound.finance/governance}},
  note = {Accessed: 2023-05}
}

@misc{aaveGov,
key={aaveGov},
  title = {Aave Finance Governance},
  howpublished = {\url{https://app.aave.com/governance/}},
  note = {Accessed: 2023-05}
}

@misc{aave-report-2023,
  author = {Gauntlet},
  title = {ARFC Aave V3 Interest Rate Curve Recommendations from Gauntlet},
  year = {2023},
  url = {https://governance.aave.com/t/arfc-aave-v3-interest-rate-curve-recommendations-from-gauntlet-2023-04-27/12921},
  note = {Accessed: 2024-05-20}
}

@misc{baude2025optimalriskawareratesdecentralized,
      title={Optimal risk-aware interest rates for decentralized lending protocols}, 
      author={Bastien Baude and Damien Challet and Ioane Muni Toke},
      year={2025},
      eprint={2502.19862},
      archivePrefix={arXiv},
      primaryClass={q-fin.MF},
      url={https://arxiv.org/abs/2502.19862}, 
}

@misc{morpho2023interest,
  title={Interest Rate Model – Morpho Docs},
  author={{Morpho Labs}},
  year={2023},
  howpublished={\url{https://docs.morpho.org/overview/concepts/irm/}},
  note={Accessed: 2025-05-27}
}

@misc{halpern2024fairratesimpossiblelending,
      title={Fair Interest Rates Are Impossible for Lending Pools: Results from Options Pricing}, 
      author={Joe Halpern and Rafael Pass and Aditya Saraf},
      year={2024},
      eprint={2410.11053},
      archivePrefix={arXiv},
      primaryClass={cs.CE},
      url={https://arxiv.org/abs/2410.11053}, 
}

@article{brace1997market,
  title={The market model of interest rate dynamics},
  author={Brace, Alastair and Gatarek, Dariusz and Musiela, Marek},
  journal={Mathematical Finance},
  volume={7},
  number={2},
  pages={127--155},
  year={1997},
  publisher={Wiley Online Library}
}

@misc{pendle_amm_docs,
  title        = {AMM | Pendle Documentation},
  author       = {{Pendle Finance}},
  year         = {2023},
  url          = {https://docs.pendle.finance/ProtocolMechanics/LiquidityEngines/AMM},
  note         = {Accessed: 2025-05-27}
}

@misc{tycho,
  title = {Tycho Aggregator},
  howpublished = {\url{https://www.propellerheads.xyz/tycho}},
  note = {Accessed: 2025-05},
author = {Propeller Heads}
}

@misc{milionis2023automated,
      title={Automated Market Making and Arbitrage Profits in the Presence of Fees}, 
      author={Jason Milionis and Ciamac C. Moallemi and Tim Roughgarden},
      year={2023},
      eprint={2305.14604},
      archivePrefix={arXiv},
      primaryClass={q-fin.MF}
}

@misc{eigenlayer2025,
  title     = {Restaking Overview},
  author    = {{EigenLayer}},
  year      = {2025},
  url       = {https://docs.eigenlayer.xyz/restakers/concepts/overview},
  note      = {Accessed: 2025-05-27}
}

@misc{babylon2025,
  title     = {Bitcoin Staking on Babylon},
  author    = {{Babylon Labs}},
  year      = {2025},
  url       = {https://babylonlabs.io/},
  note      = {Accessed: 2025-05-27}
}

@misc{meshsecurity2023,
  title     = {Mesh Security Sees Itself as New Protection on the Cosmos Network},
  author    = {Hill, Bralon},
  year      = {2023},
  url       = {https://crypto.news/mesh-security-sees-itself-as-new-protection-on-the-cosmos-network/},
  note      = {Accessed: 2025-05-27}
}

@misc{lido2025,
  title     = {Lido Liquid Staking},
  author    = {{Lido Finance}},
  year      = {2025},
  url       = {https://lido.fi/},
  note      = {Accessed: 2025-05-27}
}

@misc{rocketpool2025,
  title     = {Rocket Pool - Decentralised Ethereum Liquid Staking Protocol},
  author    = {{Rocket Pool}},
  year      = {2025},
  url       = {https://rocketpool.net/},
  note      = {Accessed: 2025-05-27}
}

@misc{lidoslashing2023,
  title     = {Lido Finance Discloses 20 Slashing Events Due to Launchnodes Validator},
  author    = {{Cointelegraph}},
  year      = {2023},
  url       = {https://cointelegraph.com/news/lido-finance-launchnodes-validator-slashed},
  note      = {Accessed: 2025-05-27}
}

@misc{neuder2024risks,
  author       = {Mike Neuder and Tarun Chitra},
  title        = {The Risks of LRTs},
  year         = {2024},
  month        = {February},
  howpublished = {\url{https://ethresear.ch/t/the-risks-of-lrts/18799}},
  note         = {Posted on Ethereum Research Forum},
}

@inproceedings{tzinas2023principal,
  author    = {Tzinas, Apostolos and Zindros, Dionysis},
  title     = {The Principal--Agent Problem in Liquid Staking},
  booktitle = {Financial Cryptography and Data Security},
  series    = {Lecture Notes in Computer Science},
  volume    = {13953},
  pages     = {456--469},
  year      = {2023},
  publisher = {Springer},
  doi       = {10.1007/978-3-031-48806-1_29},
  url       = {https://link.springer.com/chapter/10.1007/978-3-031-48806-1_29}
}

@inproceedings{Angeris_2020,
	doi = {10.1145/3419614.3423251},
  
	url = {https://doi.org/10.1145%2F3419614.3423251},
  
	year = 2020,
	month = {oct},
  
	publisher = {{ACM}
},
  
	author = {Guillermo Angeris and Tarun Chitra},
  
	title = {Improved Price Oracles},
  
	booktitle = {Proceedings of the 2nd {ACM} Conference on Advances in Financial Technologies}
}

@misc{evans2020liquidity,
      title={Liquidity Provider Returns in Geometric Mean Markets}, 
      author={Alex Evans},
      year={2020},
      eprint={2006.08806},
      archivePrefix={arXiv},
      primaryClass={q-fin.MF}
}

@misc{angeris2020does,
      title={When does the tail wag the dog? Curvature and market making}, 
      author={Guillermo Angeris and Alex Evans and Tarun Chitra},
      year={2020},
      eprint={2012.08040},
      archivePrefix={arXiv},
      primaryClass={q-fin.TR}
}

@article{nadkarni2024adaptive,
  title={Adaptive Curves for Optimally Efficient Market Making},
  author={Nadkarni, Viraj and Kulkarni, Sanjeev and Viswanath, Pramod},
  journal={arXiv preprint arXiv:2406.13794},
  year={2024}
}

@article{nadkarni2023zeroswap,
  title={ZeroSwap: Data-driven Optimal Market Making in DeFi},
  author={Nadkarni, Viraj and Hu, Jiachen and Rana, Ranvir and Jin, Chi and Kulkarni, Sanjeev and Viswanath, Pramod},
  journal={arXiv preprint arXiv:2310.09413},
  year={2023}
}

@misc{cowswap,
  title = {CowSwap Docs},
  howpublished = {\url{https://docs.cow.fi/overview/coincidence-of-wants}},
  note = {Accessed: 2023-09},
author = {CoW protocol}
}

@misc{milionis2023myersonian,
      title={A Myersonian Framework for Optimal Liquidity Provision in Automated Market Makers}, 
      author={Jason Milionis and Ciamac C. Moallemi and Tim Roughgarden},
      year={2023},
      eprint={2303.00208},
      archivePrefix={arXiv},
      primaryClass={cs.GT}
}

@misc{ethereum_staking,
  title={Ethereum Staking – Ethereum.org},
  author={{Ethereum Foundation}},
  year={2023},
  url={https://ethereum.org/en/staking/},
  note={Accessed: 2025-05-27}
}

@book{oksendal2003stochastic,
  title={Stochastic Differential Equations: An Introduction with Applications},
  author={{\O}ksendal, Bernt},
  edition={6th},
  year={2003},
  publisher={Springer},
  address={Berlin, Heidelberg},
  isbn={9783540047582}
}

@misc{libor_bba,
  title={Understanding LIBOR},
  author={{British Bankers’ Association}},
  year={2012},
  url={https://www.fca.org.uk/markets/benchmarks/libor},
  note={Accessed: 2025-05-27}
}
